\documentclass[11pt]{article}
\usepackage{amsthm}
\usepackage{graphicx} 
\usepackage{array} 
\usepackage{mathtools}
\usepackage{amsmath, amssymb, amsfonts, verbatim,bbm}
\usepackage{hyphenat,epsfig,subcaption,multirow}
\usepackage{nicefrac}
\usepackage{paralist}
\usepackage{tablefootnote}
\usepackage[utf8]{inputenc}
\usepackage[font=small, labelfont=bf]{caption}

\usepackage[usenames,dvipsnames]{xcolor}

\DeclareFontFamily{U}{mathx}{\hyphenchar\font45}
\DeclareFontShape{U}{mathx}{m}{n}{
      <5> <6> <7> <8> <9> <10>
      <10.95> <12> <14.4> <17.28> <20.74> <24.88>
      mathx10
      }{}
\DeclareSymbolFont{mathx}{U}{mathx}{m}{n}
\DeclareMathSymbol{\bigtimes}{1}{mathx}{"91}
\usepackage{tcolorbox}
\tcbuselibrary{skins,breakable}
\tcbset{enhanced jigsaw}

\usepackage[normalem]{ulem}
\usepackage[compact]{titlesec}

\definecolor{DarkRed}{rgb}{0.1,0.1,0.8}
\definecolor{DarkBlue}{rgb}{0.1,0.1,0.5}

\usepackage{nameref}
\definecolor{ForestGreen}{rgb}{0.1333,0.5451,0.1333}
\definecolor{Red}{rgb}{0.9,0,0}
\usepackage[linktocpage=true,
	pagebackref=true,colorlinks,
		urlcolor=black,
	linkcolor=DarkRed, citecolor=ForestGreen,
	bookmarks,bookmarksopen,bookmarksnumbered]
	{hyperref}
\usepackage[noabbrev,nameinlink]{cleveref}
\crefname{property}{property}{Property}
\creflabelformat{property}{(#1)#2#3}
\crefname{equation}{eq}{Eq}
\creflabelformat{equation}{(#1)#2#3}

\usepackage{bm}
\usepackage{url}
\usepackage{xspace}
\usepackage[mathscr]{euscript}

\usepackage{tikz}
\usetikzlibrary{arrows}
\usetikzlibrary{arrows.meta}
\usetikzlibrary{shapes}
\usetikzlibrary{backgrounds}
\usetikzlibrary{positioning}
\usetikzlibrary{decorations.markings}
\usetikzlibrary{patterns}
\usetikzlibrary{calc}
\usetikzlibrary{fit}
\usetikzlibrary{decorations.pathmorphing}

\usepackage{mdframed}

\usepackage[noend]{algpseudocode}
\makeatletter
\def\BState{\State\hskip-\ALG@thistlm}
\makeatother

\usepackage{cite}
\usepackage{enumitem}

\usepackage[margin=1in]{geometry}

\newtheorem{theorem}{Theorem}
\newtheorem{lemma}{Lemma}[section]

\newtheorem{proposition}[lemma]{Proposition}
\newtheorem{corollary}[lemma]{Corollary}

\newtheorem{definition}[lemma]{Definition}

\newtheorem*{claim*}{Claim}
\newtheorem*{proposition*}{Proposition}
\newtheorem*{lemma*}{Lemma}
\newtheorem*{problem*}{Problem}

\crefname{lemma}{Lemma}{Lemmas}
\crefname{claim}{Claim}{Claims}

\newtheorem{mdresult}{Result}

\newtheorem{remark}[lemma]{Remark}
\newtheorem{assumption}{A\hspace{-1mm}}
\newtheorem{observation}[lemma]{Observation}

\newtheoremstyle{restate}{}{}{\itshape}{}{\bfseries}{~(restated).}{.5em}{\thmnote{#3}}
\theoremstyle{restate}

\theoremstyle{definition}
\newtheorem{mdalg}{Algorithm}

\allowdisplaybreaks

\renewcommand{\qed}{\nobreak \ifvmode \relax \else
      \ifdim\lastskip<1.5em \hskip-\lastskip
      \hskip1.5em plus0em minus0.5em \fi \nobreak
      \vrule height0.75em width0.5em depth0.25em\fi}

\newcommand{\eps}{\ensuremath{\epsilon}}

\newcommand{\paren}[1]{\ensuremath{\left(#1\right)}\xspace}

\newcommand{\IR}{\ensuremath{\mathbb{R}}}

\newcommand{\prob}[1]{\Pr\paren{#1}}

\DeclareMathOperator*{\Exp}{\ensuremath{{\mathbb{E}}}}
\DeclareMathOperator*{\Prob}{\ensuremath{\mathbb{P}}}
\renewcommand{\Pr}{\Prob}

\newenvironment{tbox}{\begin{tcolorbox}[
		enlarge top by=5pt,
		enlarge bottom by=5pt,
		 breakable,
		 boxsep=0pt,
                  left=4pt,
                  right=4pt,
                  top=10pt,
                  arc=0pt,
                  boxrule=1pt,toprule=1pt,
                  colback=white
                  ]
	}
{\end{tcolorbox}}

\newcommand{\II}{\ensuremath{\mathbb{I}}}

\newcommand{\mireal}[1][]{
  \ifx\relax#1\relax%
    \II(\mione \,; \mitwo)%
  \else%
    \II(\mione \,; \mitwo\mid #1)%
  \fi
}

\newcommand{\cS}{\mathcal{S}}

\newcommand{\ip}[2]{\ensuremath{\langle  #1, #2 \rangle}}
\newcommand{\norm}[1]{\ensuremath{\left\lVert #1 \right\rVert}}
\newcommand{\A}[1]{\ensuremath{W^{(#1)}}}
\newcommand{\Em}[1]{\ensuremath{\mathcal{E}^{(#1)}}}

\newcommand{\um}[1]{\ensuremath{u^{(#1)}}}
\newcommand{\vm}[1]{\ensuremath{v^{(#1)}}}
\newcommand{\sm}[1]{\ensuremath{s^{(#1)}}}

\newcommand{\infnorm}[1]{\ensuremath{\left\lVert #1 \right\rVert}_{\infty}}
\newcommand{\twonorm}[1]{\ensuremath{\left\lVert #1 \right\rVert}_{2}}
\newcommand{\twotoinfnorm}[1]{\ensuremath{\left\lVert #1 \right\rVert}_{2\rightarrow \infty}}

\newcommand{\hsbm}{\textnormal{HSBM}}

\newcommand{\sign}{\textnormal{sgn}}

\usepackage{algorithm}
\usepackage{algpseudocode}

\makeatletter
\newenvironment{breakablealgorithm}
  {
   \begin{center}
     \refstepcounter{algorithm}
     \hrule height.8pt depth0pt \kern2pt
     \renewcommand{\caption}[2][\relax]{
       {\raggedright\textbf{\ALG@name~\thealgorithm} ##2\par}%
       \ifx\relax##1\relax 
         \addcontentsline{loa}{algorithm}{\protect\numberline{\thealgorithm}##2}%
       \else 
         \addcontentsline{loa}{algorithm}{\protect\numberline{\thealgorithm}##1}%
       \fi
       \kern2pt\hrule\kern2pt
     }
  }{
     \kern2pt\hrule\relax
   \end{center}
  }
\makeatother

\definecolor{applegreen}{rgb}{0.55, 0.71, 0.0}

\allowdisplaybreaks

\title{Community Detection in the Hypergraph SBM: \\Exact Recovery Given the Similarity Matrix \footnotetext{Accepted for presentation at the Conference on Learning Theory (COLT) 2023.}
}
\author{Julia Gaudio \\
\href{mailto:julia.gaudio@u.northwestern.edu}{\text{julia.gaudio@northwestern.edu}} \\
Northwestern University. \and 
Nirmit Joshi\\
\href{mailto:nirmit@u.northwestern.edu}{\text{nirmit@ttic.edu}}\\
TTI-Chicago}
\date{}
\begin{document}
\maketitle

\begin{abstract}
 Community detection is a fundamental problem in network science. In this paper, we consider community detection in hypergraphs drawn from the \emph{hypergraph stochastic block model} (HSBM), with a focus on exact community recovery. We study the performance of polynomial-time algorithms which operate on the \emph{similarity matrix} $W$, where $W_{ij}$ reports the number of hyperedges containing both $i$ and $j$. Under this information model, while the precise information-theoretic limit is unknown, Kim, Bandeira, and Goemans derived a sharp threshold up to which the natural  min-bisection estimator on $W$ succeeds. As min-bisection is NP-hard in the worst case, they additionally proposed a semidefinite programming (SDP) relaxation and conjectured that it achieves the same recovery threshold as the min-bisection estimator.
  
 In this paper, we confirm this conjecture. We also design a simple and highly efficient spectral algorithm with nearly linear runtime and show that it achieves the min-bisection threshold. Moreover, the spectral algorithm also succeeds in denser regimes and is considerably more efficient than previous approaches, establishing it as the method of choice. Our analysis of the spectral algorithm crucially relies on strong \emph{entrywise} bounds on the eigenvectors of $W$. Our bounds are inspired by the work of Abbe, Fan, Wang, and Zhong, who developed entrywise bounds for eigenvectors of symmetric matrices with independent entries. Despite the complex dependency structure in similarity matrices, we prove similar entrywise guarantees.
\end{abstract}

\section{Introduction}
Community detection is the problem of partitioning a network into densely connected clusters. As a fundamental network science problem, community detection arises in numerous applications: sociology \cite{10.1561/2200000005, Newman-randomgraph}, protein interactions \cite{10.1093/bioinformatics/btl370,doi:10.1126/science.285.5428.751}, image applications \cite{868688}, natural language processing \cite{PhysRevE.84.036103}, webpage sorting \cite{KUMAR19991481} and many more. In 1983, \cite{Holland1983} introduced the \emph{stochastic block model} (SBM), a probabilistic generative model for networks with community structure. Since then, community detection in the SBM has been intensely studied in the probability, statistics, and theoretical computer science communities \cite{Dyer1989,McSherry2001,Decelle2011,https://doi.org/10.48550/arxiv.1311.4115,10.1145/2591796.2591857,7354421,Abbe2016ExactRI,NIPS2015-cfee3986}; also see \cite{Abbe2017} for a survey.

In this paper, we consider an extension of the SBM to \emph{hypergraphs}. A hypergraph is a generalization of a graph that captures higher-order interactions.  For example, an academic co-authorship network may be modeled as a hypergraph, where each hyperedge represents the author list of a paper. Formally, a hypergraph is specified by a set of vertices $V$ and a set of hyperedges $E$. Each hyperedge $e \in E$ is a subset of $V$. We specialize to \emph{uniform} hypergraphs, where each hyperedge contains the same number of vertices. A \emph{$d$-uniform hypergraph} satisfies $|e| = d$ for all $e \in E$ (in particular, a graph is $2$-uniform hypergraph). We then say that the hypergraph has \emph{order} $d$. 

We now describe the Hypergraph Stochastic Block Model (HSBM), which was first proposed by \cite{Ghoshdastidar2014}. We consider the version with two balanced communities and equal inter-community edge probabilities. The model is specified by its order $d$ and two parameters $1 \geq p_n> q_n >0$. First, a community assignment vector $\sigma^{*} \in \{\pm 1\}^n$ is sampled uniformly at random from the set $\{\sigma \in \{\pm 1\}^n : \langle \mathbf{1}_n, \sigma\rangle = 0\}$\footnote{For simplicity of exposition, we assume $n$ is even, and thus, each community has exactly $n/2$ vertices.}, where $\mathbf{1}_n \in \mathbb{R}^n$ is the vector of all ones. 

Conditioned on $\sigma^{*}$, we sample a hypergraph $G = ([n], E)$ as follows. Each $e= \{i_1, i_2, \dots, i_d\} \in \binom{[n]}{d}$ appears as a hyperedge independently with probability
\begin{align*}
\prob{e \in E} = \begin{cases}
p_n & \sigma^{*}(i_1) = \sigma^{*}(i_1) = \cdots = \sigma^{*}(i_d)\\
q_n & \text{otherwise}.
\end{cases}
\end{align*}
We then write $G \sim \text{HSBM}(d,n, p_n, q_n)$. Throughout, we consider $d$ to be a constant and denote $\mathcal{E}:=\binom{[n]}{d}$ to be the set of all possible hyperedges. We use the parametrization:
\begin{equation}\label{eq:parameter-regimes}
    p_n=\alpha f_n \hspace{3mm} \text{and} \hspace{3mm}  q_n=\beta f_n;
\end{equation}
\begin{center}
   where either ($f_n=o(1)$ and $\alpha>\beta>0$) or ($f_n=1$ and $1 \geq \alpha>\beta>0$) for constants $\alpha,\beta$. 
\end{center}
 We are interested in algorithms that recover all of the vertex labels. More formally, we say that an estimator $\hat{\sigma}_n$ \emph{achieves exact recovery} if
\[\lim_{n \to \infty} \prob{\hat{\sigma}_n \in \{\pm \sigma^{*}_n\}} = 1.\]
For clarity of presentation, we typically drop the dependence on $n$. 

The limiting regime for the exact recovery problem is $f_n = \Theta\left(\nicefrac{\log n}{n^{d-1}} \right)$. That is, when $f_n = o\left(\nicefrac{\log n}{n^{d-1}} \right)$, exact recovery is not possible statistically. This is because the hypergraph will have isolated vertices with a high probability, which are impossible to classify information theoretically. On the other hand, when $f_n = \omega\left(\nicefrac{\log n}{n^{d-1}}\right)$, there are efficient algorithms for exact recovery \cite{Chien2019}. In the logarithmic degree regime, we typically parametrize as $f_n = \nicefrac{\log n}{\binom{n-1}{d-1}}$. 
In the corresponding model $\text{HSBM}(d,n,\alpha f_n, \beta f_n)$, there is a precise information-theoretic threshold, determined by \cite{Kim2018}.
If $I_{\textnormal{full}}(d,\alpha,\beta):=\frac{1}{2^{d-1}} \left(\sqrt{\alpha} - \sqrt{\beta} \right)^2 < 1$, then exact recovery is impossible, while if $I_{\textnormal{full}} (d,\alpha,\beta) > 1$, then there are polynomial-time algorithms for the exact recovery problem developed by \cite{Ghoshdastidar2015a,Ghoshdastidar2015b,Ghoshdastidar2017,Ahn2018,Chien2019}, culminating in the work of \cite{Zhang2021}, whose algorithm applies to a general class of HSBMs. Their results essentially show that there is no statistical-computational gap for the exact recovery problem.

While the work of \cite{Zhang2021} resolves the question of poly-time exact recovery even for general HSBMs, the algorithm uses the full hypergraph information. Unfortunately, storing the full hypergraph information can be prohibitively expensive; in the regime where $p_n, q_n = \Theta(1)$, it requires $\Theta(n^d)$ space. A natural question arises: is there some other way of storing the hypergraph information which uses less space than storing the full information, while still being powerful enough for exact recovery? One candidate information model is the so-called \emph{similarity matrix}.
\begin{definition}[\textbf{Similarity Matrix}]\label{def:similarity-matrix}
Let $G = ([n], E)$ be a hypergraph on $n$ vertices. The similarity matrix of $G$ is the zero-diagonal matrix $W$ whose entries are
\[W_{ij} = \left|\left\{e \in E : \{i,j\} \subset e  \right\} \right|\]
for $i \neq j$. In other words, $W_{ij}$ counts the number of edges which contain both $i$ and $j$. We also write $W = \mathcal{S}(G)$ to define the similarity matrix transformation.
\end{definition}
Even in the case when $p_n,q_n = \Theta(1)$, the similarity matrix requires only $O(n^2 \log n)$ space to store. Recent works \cite{Lee2020,Kim2018,cole2020exact} have considered the exact recovery problem, given the similarity matrix $W = \mathcal{S}(G)$, where $G \sim \text{HSBM}(d, n, p_n, q_n)$. \cite{Lee2020} showed that the asymptotic regime under which the similarity matrix is powerful enough for exact recovery is given by $p_n - q_n = \Omega\left(\sqrt{p_n \log n / n^{d-1}} \right)$. In the logarithmic degree regime where $f_n=\log n/\binom{n-1}{d-1}$, while the precise information-theoretic threshold (in terms of $d, \alpha ,\beta$) for exact recovery given $W=\mathcal{S}(G)$ remains unknown\footnote{A previous version of this work claimed that $I(d,\alpha,\beta) = 1$ is the information-theoretic threshold, when given the similarity matrix. However, \cite[Theorem 3]{Kim2018} establishes $I(d,\alpha,\beta) = 1$ as the threshold for the min-bisection estimator, and it is not known whether this threshold coincides with the information-theoretic threshold given the similarity matrix.}, 
\cite[Theorem 3]{Kim2018} analyzed the performance of a natural min-bisection estimator.  
Letting \begin{equation}
I(d, \alpha, \beta) = \max_{t \geq 0} \frac{1}{2^{d-1}} \left[\alpha \left(1 - e^{-(d-1)t} \right) + \beta \sum_{r=1}^{d-1}\binom{d-1}{r} \left(1 - e^{-(d-1-2r)t}\right)  \right],
\end{equation}
the min-bisection estimator on $W$ achieves exact recovery if $I(d, \alpha, \beta) > 1$ and fails if $I(d,\alpha,\beta)<1$. Note that the min-bisection problem is NP-hard in the worst-case. Therefore, \cite{Kim2018} additionally proposed a semidefinite programming (SDP) relaxation, showing that it succeeds in exact recovery if $I_{\text{SDP}}(d, \alpha, \beta) > 1$, where
\vspace{-1mm} \[I_{\text{SDP}}(d, \alpha, \beta) = \frac{(d-1)}{2^{2d}} \cdot \frac{(\alpha - \beta)^2}{\left(\alpha \frac{d}{2^d} + \beta \left(1 - \frac{d}{2^d} \right)\right)}.\]
Note that $I_{\text{SDP}}(d, \alpha, \beta) > I(d, \alpha, \beta)$ for $d \geq 3$; see Figure \ref{fig:phase-transition} for an illustration. The authors conjectured that the SDP achieves the min-bisection threshold; i.e. it succeeds whenever $I(d,\alpha, \beta) > 1$ \cite[Conjecture 1.2]{Kim2018}. This leaves an open question: 
\begin{center}
    \emph{Can this sharp threshold for min-bisection be achieved by a polynomial time algorithm that only uses the similarity matrix?}
\end{center} 
We also investigate the question of whether using the similarity matrix enables efficient, nearly linear runtime algorithms. Motivated by the success of spectral algorithms for other community detection problems \cite{Abbe2020ENTRYWISEEA, Deng2021, Dhara2022a, Dhara2022b}, and given their efficiency, we ask:
\begin{center}
\emph{Can we design a spectral algorithm with nearly linear runtime that achieves the min-bisection threshold from the similarity matrix?}    
\end{center}
\subsection{Our Contributions} 
Our first main contribution is to show that the SDP algorithm succeeds whenever $I(d,\alpha,\beta)>1$ (Theorem \ref{theorem:SDP}), achieving the min-bisection threshold. We note that there has been some follow-up work on community recovery from the similarity matrix \cite{Lee2020, cole2020exact}, but to our knowledge, we are the first to resolve \cite[Conjecture 1.2]{Kim2018}. We also show that the SDP is robust to monotone adversarial changes on the similarity matrix.

Our second main contribution is a simple spectral algorithm based on the similarity matrix that also achieves the min-bisection threshold (Theorem \ref{theorem:spectral}).
\begin{figure}
    \centering
    \includegraphics[width=100mm, height=75mm, scale=2.0]{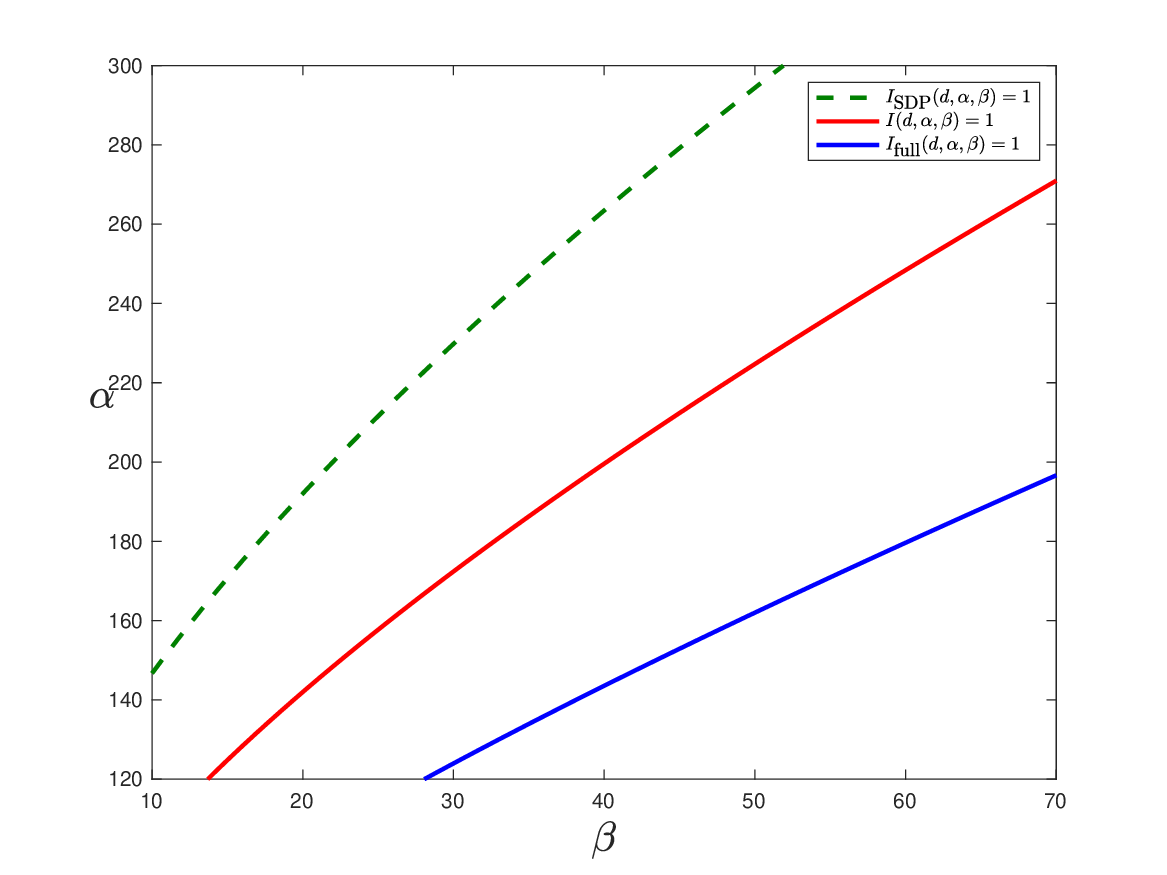}
    \caption{Visualization of $I_{\textnormal{SDP}}(d,\alpha,\beta)=1$, $I(d,\alpha,\beta)=1$ and $I_{\textnormal{full}}(d,\alpha,\beta)=1$ when $d=6$.}
    \label{fig:phase-transition}
\end{figure}
Our algorithm determines communities based on the signs of the entries of the second eigenvector of the similarity matrix and does not require any clean-up. Furthermore, the algorithm can be implemented in $O(n \log^2 n)$ time in the logarithmic degree regime, using the fast eigenvector computation method of \cite{garber2016faster}. When $f_n = \omega(\log n /n^{d-1})$ and full hyperedge information is available, we may subsample the hypergraph to return to the logarithmic degree regime. Since the subsampling procedure takes $O(n \log n)$ time, the overall runtime is still $O(n \log^2 n)$. Table \ref{tab:run-time-full_information} summarizes the runtime of the algorithm, compared to other approaches in the literature (allowing for subsampling), including results coming after our initial posting on arXiv. These works include \cite{wang2023projected}, which proposed the Projected Tensor Power Method. The algorithm achieves the exact recovery threshold $I_{\textnormal{full}} (d,\alpha,\beta) > 1$, with a runtime of $O(n \log^2(n) / \log \log n)$, when initialized from a labeling satisfying a certain partial recovery condition \cite[Theorem 4.1]{wang2023projected}. Most recently, \cite{Dumitriu2023} proposed a two-stage algorithm achieving exact recovery down to the information-theoretic threshold. \cite{Dumitriu2023} also identified the recovery threshold in the non-uniform case, and showed that their algorithm achieves the information-theoretic threshold in that case also.

We also compare runtimes when only the similarity matrix is available. In this situation, we cannot subsample down to the logarithmic degree regime, 
 due to the loss of hyperedge information. Nevertheless, we show that our spectral algorithm achieves exact recovery directly from $W=\mathcal{S}(G)$, where $G\sim \text{HSBM}(d,n,p_n,q_n)$, in all the regimes of $p_n, q_n$ captured by \eqref{eq:parameter-regimes} as long as $f_n=\omega(\log n/n^{d-1})$ (Theorem \ref{theorem:spectral-otheregimes}). Moreover, we show that its runtime is $O(n^2 \log n)$ in the worst-case; i.e. $p_n,q_n=\Theta(1)$. See Table \ref{tab:run-time-similarity-matrix}.
 
\begin{table}[!htp]
\begin{center}
\begin{tabular}{ |c|c| } 
 \hline
 Strategy (Reference) & Runtime (all regimes) \\
 \hline
 \hline
 Spectral (This work) & $O(n \log^2 n)$  \\ 
 \hline
 SDP \cite{Kim2018} & $\tilde{O}(n^{3.5})$ \cite{Jiang2020} \\ 
 \hline
 Spectral+refinement \cite{Chien2019} & $O(n^3 \log n)$ \\ 
 \hline
  Projected Tensor Power Method \cite{wang2023projected} & $O(n \log^2(n))$ \tablefootnote{Given a suitable initialization that satisfies a certain partial recovery criterion, their algorithm runs in $O(\frac{n \log^2 n}{\log \log n})$ time; see \cite[Theorem 4.1]{wang2023projected}. However, to the best of our knowledge, this criterion can only be achieved in $O(n \log^2 n)$ time using current methods \cite{Dumitriu2023}.} \\ 
 \hline
Spectral + refinement \cite{Dumitriu2023} & $O(n \log^2(n))$  \\ 
 \hline
\end{tabular}
\caption{Runtime comparison when the full hypergraph is known.} \label{tab:run-time-full_information}
\vspace{-4mm}
\end{center}
\end{table}

\begin{table}[ht]
\begin{center}
\begin{tabular}{ |c|c|c| } 
 \hline
 Strategy & Runtime & Runtime\\
  (Reference) & (Logarithmic degree regime) & (All regimes)\\
 \hline
 \hline
 Spectral (This work) & $O(n \log^2 n)$ & $O(n^2 \log n)$ \\ 
 \hline
 SDP \cite{Kim2018} & \multicolumn{2}{c|}{$\tilde{O}(n^{3.5})$ \cite{Jiang2020}} \\ 
 \hline
\end{tabular}
\caption{Runtime comparison when only the similarity matrix is known.}\label{tab:run-time-similarity-matrix}
\vspace{-4mm}
\end{center}
\end{table}

In order to analyze our spectral algorithm, inspired by the work of \cite{Abbe2020ENTRYWISEEA}, we develop $\ell_{\infty}$ (entrywise) bounds for the eigenvectors of similarity matrices of a large class of random hypergraph models (Theorem \ref{thm:entrywise-analysis}), which may be of independent interest. Roughly speaking, we show that an eigenvector $u_k$ of a random similarity matrix $W$ is close to its first order approximation $Wu_k^*/\lambda_k^*$ in the $\ell_\infty$ norm under mild conditions, where $(\lambda_k^*,u_k^*$) is the corresponding eigenpair of $\Exp{[W]}$. Our result addresses two important questions raised by \cite{Abbe2020ENTRYWISEEA} by: (1) providing an example of entrywise eigenvector approximation beyond symmetric matrices with independent entries, (2) expanding the class of graph-based matrices for which entrywise eigenvector guarantees are known, beyond adjacency matrices \cite{Abbe2020ENTRYWISEEA} and Laplacian matrices \cite{Deng2021}.

\paragraph*{Organization}
Section \ref{sec:results} contains our main results. We give proof outlines of the SDP and spectral algorithm results in Section \ref{sec:SDP-outline} and \ref{sec:spectral-outline}, respectively. Directions for future work are proposed in Section \ref{sec:discussion}. The proofs of our main results are provided in the appendix. 

\subsection{Further Related Work}
\paragraph{Other recovery problems.} While we only focus on exact recovery in this work, partial recovery (recovering a non-trivial constant fraction of the community labels) and almost exact recovery (recovering all but a vanishing fraction of the community labels)
have also been studied in the context of the HSBM
\cite{Angelini2015,Zhang2022,Dumitriu2021,Zhen2022,Ke2019}. Community recovery in the non-uniform HSBM has been studied \cite{Ghoshdastidar2017,Dumitriu2021,Alaluusua2023}, with sharp results in the exact recovery problem given by the very recent work of \cite{Dumitriu2023}. Additionally, \cite{Dumitriu2023} study almost exact recovery in non-uniform HSBM and show that it is possible whenever the maximal expected degree diverges, and give an efficient algorithm for the same.  

\paragraph{Spectral methods.} Spectral algorithms have been very successful in statistical inference problems, e.g. community detection \cite{newman2006finding, Rohe_2011, McSherry2001,Abbe2020ENTRYWISEEA, Dhara2022c, Deng2021}, the Planted Clique Problem \cite{alon1998finding}, clustering \cite{von2007tutorial, ng2001spectral}, dimensionality reduction \cite{belkin2003laplacian} and many more; see \cite{chen2021spectral} for a survey on the topic. 

\paragraph{Follow-up work.} 
\cite{Alaluusua2023} considered a multi-layer version of the HSBM, where each layer corresponds to an independent, regular HSBM. 
Their work provides a sufficient condition for exact recovery from the aggregated similarity matrix (the sum of the similarity matrices of all layers), using an SDP approach. 

\subsection{Notation}\label{subapp:notation}
For any real numbers $a, b \in R$, we denote $a \vee b = \max\{a, b\}$ and $a \wedge b = \min\{a, b\}$. Let $\sign: \mathbb{R} \to \{\pm 1\}$ be the function defined by
$\sign(x) = 1$ if $x \geq 0$ and $\sign(x) = -1$ if $x < 0$. We also extend the definition to vectors; let $\sign : \mathbb{R}^n \to \{\pm 1\}^n$ be the map defined by applying the sign function componentwise. 

We use the notation $\mathbb{R}_+ = (0,\infty)$. For $n \in \mathbb{N}$, we write $[n] = \{1, 2, \dots, n\}$. We use the Bachmann–Landau notation $o(.)$, $O(.)$, $\omega(.)$, $\Omega(.)$, $\Theta(.)$ etc. throughout the paper. For nonnegative sequences $(a_n)_{n\geq 1}$
and $(b_n)_{n\geq 1}$, we write $a_n \lesssim b_n$ to mean
$a_n \leq Cb_n$ for some constant $C > 0$. The notation $\asymp$ is similar, hiding two constants
in upper and lower bounds. Moreover, we write $a_n  \approx b_n$ as a shorthand for $\lim_{n\rightarrow \infty} \frac{a_n}{b_n}= 1$.

For any two vectors $x,y \in \IR^n$, $\ip{x}{y}$ represents the standard inner product in $\IR^n$; we define $\twonorm{x}=(\sum_{i=1}^{n} x_i^2)^{1/2}$, $\Vert x \Vert_1=\sum_{i=1}^{n} |x_i|$, and $\infnorm{x}=\max_{i} |x_i|$. For any matrix $M \in \IR^{n\times n}$, $M_{i\cdot}$ refers to its $i$-th row, which is a row
vector, and $M_{\cdot i}$ refers to its $i$-th column, which is a column vector. The matrix spectral
norm is $\twonorm{M} = \sup_{\twonorm{x}=1} \twonorm{Mx}$, the matrix $2 \rightarrow \infty$ norm is $\twotoinfnorm{M} = \sup_{\twonorm{x}=1} \infnorm{Mx}= \sup_{i} \twonorm{M_{i\cdot}}$, and the the matrix Frobenius norm is $\norm{M}_F=(\sum_{i=1}^{n}\sum_{j=1}^{n} M_{ij}^2)^{1/2}$.

\section{Main Results}\label{sec:results}
Recall that $\sigma^{*}_n \in \{ \pm 1\}^n$ denotes the true community assignment vector. Let $G$ be a hypergraph on $n$ vertices, and let $W = \mathcal{S}(G)$ be its similarity matrix. \cite{Kim2018} proposed an SDP for exact recovery (Algorithm \ref{alg:SDP}). Our first main result states that, in the logarithmic degree regime, the SDP relaxation achieves exact recovery whenever the min-bisection estimator does. 
\begin{theorem}\label{theorem:SDP}
Fix $d \in \{2, 3, \dots\}$ and $\alpha > \beta > 0$ such that $I(d, \alpha, \beta) > 1$. Let $f_n=\log n/\binom{n-1}{d-1}$. Suppose $G \sim \hsbm(d, n, \alpha f_n, \beta f_n)$, and let $W = \mathcal{S}(G)$. Let $\hat{X}$ be the optimal solution to \eqref{eq:SDP} with input $W$. Then $\hat{X} = \sigma^* {\sigma^*}^{\top}$ with probability $1-o(1)$. It follows that
$$ \lim_{n \rightarrow \infty} \prob{ \hat{\sigma}_{\textnormal{\tiny{SDP}}} \in \{ \pm \sigma^* \} } =1.$$
\end{theorem}
\begin{breakablealgorithm}
\caption{SDP recovery algorithm \cite{Kim2018}}\label{alg:SDP}
\begin{algorithmic}[1]
\Require{An $n \times n$ similarity matrix $W$}
\vspace{0.2cm}
\Ensure{An estimate of community assignments}
\vspace{0.2cm}
\State Solve the following SDP, where $X \in \mathbb{R}^{n \times n}$.
\vspace{-2mm}
\begin{equation}
\begin{aligned}
\max~&\langle W, X \rangle\\
\text{subject to } & X_{ii} = 1 \text{ for all } i \in [n]\\
&\langle X, \mathbf{11}^\top\rangle = 0,\\
&  X \succeq 0.
\end{aligned} \label{eq:SDP}
\end{equation}
\vspace{-2mm}
\State Let $\hat{X}$ be the optimal solution, and let $\hat{X} = \sum_{i=1}^n \lambda_i v_i v_i^\top$ denote the eigendecomposition of $\hat{X}$, where $\lambda_1 \geq \lambda_2 \geq \cdots \geq \lambda_n$.
\State Return $\hat{\sigma}_{ \text{SDP}} = \sign (v_1)$.
\end{algorithmic}
\end{breakablealgorithm}
\noindent Note that the SDP algorithm also works in denser regimes of \eqref{eq:parameter-regimes} \cite{Lee2020}. We also establish that the SDP continues to achieve exact recovery even under a \emph{monotone adversary} model. 
\begin{lemma}\label{lemma:monotone-adversary}
Consider the modified SDP based on \eqref{eq:SDP}, with $W$ replaced by $\widetilde{W}$ such that $\widetilde{W}_{ij} \geq W_{ij}$ if $\sigma^*(i)=\sigma^*(j)$, and $\widetilde{W}_{ij} \leq W_{ij}$ if $\sigma^*(i) \neq \sigma^*(j)$. If $I(d, \alpha, \beta) > 1$, then $X^*:=\sigma^*{\sigma^*}^{\top}$ is the unique optimal solution to the modified SDP.
\end{lemma}

Having analyzed the performance of the SDP relaxation, we now propose a spectral algorithm.
\begin{breakablealgorithm}
\caption{Spectral recovery algorithm}\label{alg:spectral}
\begin{algorithmic}[1]
\Require{An $n \times n$ similarity matrix $W$}
\vspace{.2cm}
\Ensure{An estimate of community assignments}
\vspace{0.2cm}
\State Compute the second eigenpair of $W$, denoted by $(\lambda_2, u_2)$, where $\lambda_1 \geq \lambda_2 \geq \cdots \geq \lambda_n$.
\State Return $\hat{\sigma}_{ \text{spec}} = \sign (u_2)$.
\end{algorithmic}
\end{breakablealgorithm}
We first establish Algorithm \ref{alg:spectral} achieves exact reovery up to the min-bisection threshold.
\begin{theorem}\label{theorem:spectral}
Fix $d \in \{2, 3, \dots\}$ and $\alpha > \beta > 0$ such that $I(d, \alpha, \beta)>1$. Let $f_n=\log n/\binom{n-1}{d-1}$. Suppose $G \sim \hsbm(d, n, \alpha f_n, \beta f_n)$, and let $W = \mathcal{S}(G)$. Let $u_2$ be the second eigenvector of $W$. Then with probability $1-o(1)$, there exist $s \in \{\pm 1\}$ and $\eta = \eta(d, \alpha, \beta) > 0$ such that 
\[ \sqrt{n} \min_{i \in [n]} s \sigma^*(i) u_{2,i} >\eta.\] 
As a result, the estimator $\hat{\sigma}_{\textnormal{\tiny{spec}}}$ produced by Algorithm \ref{alg:spectral} on input $W$ achieves exact recovery.
\end{theorem}
We also show that the spectral algorithm succeeds in all the super-logarithmic degree regimes in \eqref{eq:parameter-regimes}.
\begin{theorem}\label{theorem:spectral-otheregimes}
Fix $d \in \{2, 3, \dots\}$. Let $p_n$ and $q_n$ be parameterized according to \eqref{eq:parameter-regimes} for some $f_n$ and constants $\alpha> \beta>0$. Suppose $G \sim \hsbm(d, n, \alpha f_n, \beta f_n)$, and let $W = \mathcal{S}(G)$. If $f_n=\omega( \log n/n^{d-1})$, then the estimator $\hat{\sigma}_{\textnormal{\tiny{spec}}}$ produced by Algorithm \ref{alg:spectral} on input $W$ achieves exact recovery; i.e. $$ \lim_{n \rightarrow \infty } \prob{\hat{\sigma}_{\textnormal{\tiny{spec}}} \in \{ \pm \sigma^* \}} =1.$$
\end{theorem}
Our proofs of Theorem \ref{theorem:spectral} and \ref{theorem:spectral-otheregimes} crucially rely on \emph{entrywise} bounds on the second eigenvector of the similarity matrix. To this end, we develop entrywise bounds on the eigenvectors of similarity matrices of a generic family of random hypergraphs (Definition \ref{def:general-hypergraph}).  
\begin{definition}[\textbf{General random hypergraph}]\label{def:general-hypergraph}
Let $d \in \{2,3, \dots\}$, $n \in \mathbb{N}$, and $p \in [0,1]^{\binom{[n]}{d}}$. Define $H(d,n,p)$ to be the distribution on $d$-uniform hypergraphs with $n$ vertices, where each edge $e \in \binom{[n]}{d}$ appears in the hypergraph with probability $p_e$, independently.
\end{definition}

We analyze the eigenvectors of $W = \mathcal{S}(G)$, where $G \sim H(d,n,p)$. Let $(\lambda_i, u_i)_{i=1}^n$ denote the eigenpairs of $W$, where $\lambda_1 \geq \lambda_2 \geq \cdots \geq \lambda_n$. Let $W^{*} = \mathbb{E}[W]$, with ordered eigenpairs $(\lambda_i^{*}, u_i^{*})_{i=1}^n$. We use the convention $\lambda_0 = \lambda_0^* = +\infty$ and $\lambda_{n+1} = \lambda_{n+1}^* = -\infty$. We then define the following eigengap quantity:
\[\Delta_k^* = \min\{\lambda_{k-1}^*-\lambda_k^* , \lambda_{k}^*-\lambda_{k+1}^* \}.\]
Our entrywise guarantee requires a spectral separation assumption. The assumption easily holds for similarity matrices of (general) HSBMs, in all the parameter regimes we are interested in. 
\begin{assumption}[\textbf{Spectral separation}]\label{asump:well-separated}
There is a sequence $\{\mu_n\}$ in $(0,\infty)$ such that
$$ \max\{p_e: e \in \mathcal{E}\} \leq \mu_n, \textnormal{ and} \hspace{2mm} n \binom{n-2}{d-2} \mu_n \geq c_0 \log n,$$
for some constant $c_0>0$. Moreover, there is a constant $c_1 \geq 1$ such that 
$$ \frac{1}{c_1}n^{d-1} \mu_n \leq |\lambda_k^*|, |\Delta_k^*| \leq c_1 n^{d-1} \mu_n; \textnormal{ i.e.} \hspace{2mm}|\lambda_k^*|, |\Delta_k^*|=\Theta(n^{d-1}\mu_n).$$
\end{assumption}
\noindent Under this assumption, we state our entrywise guarantee.
\begin{theorem}\label{thm:entrywise-analysis}
Let $k \in \mathbb{N}$ and $d \in \{2,3,\dots \}$ be constants. Let $p \in [0,1]^{\binom{[n]}{d}}$, such that Assumption \ref{asump:well-separated} holds for some $\mu_n$ and constants $c_0, c_1 > 0$. Let $G\sim H(d,n,p)$, and $W=\mathcal{S}(G)$. Then with probability $1-O(n^{-3})$,
$$\min_{s^*\in\{\pm 1\}} \infnorm{u_k -s^* \frac{W{u_k^*}}{\lambda_k^*}}\leq \hspace{2mm}\frac{c \infnorm{u_k^*}}{\log \log n},$$
where $c>0$ is some constant depending only on $d,c_0,$ and $c_1$. 
\end{theorem}
\begin{remark}
We remark that when $d=2$, the similarity matrix of a graph is just its adjacency matrix and our Theorem \ref{thm:entrywise-analysis} recovers the entrywise bounds of \cite{Abbe2020ENTRYWISEEA}. Moreover, $I(2,\alpha,\beta)= (\sqrt{\alpha} -\sqrt{\beta})^2/2$, which is the information-theoretic threshold for exact recovery in the SBM setting \cite{abbe2015exact}. Therefore, our Theorem \ref{theorem:SDP} and Theorem \ref{theorem:spectral}, respectively, recover \cite[Theorem 2]{Hajek2016} and \cite[Theorem 3.2]{Abbe2020ENTRYWISEEA}.
\end{remark}
\section{Analysis of the SDP Relaxation}\label{sec:SDP-outline}
\paragraph*{Analyzing the SDP.}
We use a dual certificate strategy as in \cite{Kim2018}. The dual of \eqref{eq:SDP} is given by
\begin{equation}
\begin{aligned}
\min~& \text{trace}(D) \\
\text{subject to } & D \text{ is } n\times n \text{ diagonal matrix, } \nu \in \IR,\\
& D+\nu \mathbf{11}^{\top}-W \succeq 0.
\end{aligned} \label{eq:dual}
\end{equation}
The form of the dual motivates the following sufficient condition, whose proof we include for completeness (see Appendix \ref{app:SDP-proofs}).
\begin{lemma}\label{lem:uniqueness of sdp solution}
Suppose there is a diagonal matrix $D \in \mathbb{R}^{n \times n}$ and $\nu \in \IR$ such that the following holds. Letting $S \triangleq D +\nu \mathbf{11}^\top-W $, the matrix $S$ is positive semidefinite, its second-smallest eigenvalue $\lambda_{n-1}(S)$ is strictly positive, and $S \sigma^*=0$. Then $X^*:=\sigma^*{\sigma^*}^\top$ is the unique optimal solution to \eqref{eq:SDP}.
\end{lemma}
To apply Lemma \ref{lem:uniqueness of sdp solution}, we let $D$ be the diagonal matrix whose diagonal entries are specified by
\begin{equation}
D_{ii} = \sum_{j \in [n]} W_{ij} \sigma^{*}(i) \sigma^{*}(j).   \label{eq:diagonal-matrix}
\end{equation}
Setting $\nu = 1$, write $S = D +  \mathbf{11}^\top - W$. By construction, we have $S \sigma^{*} = 0$. It remains to show
\begin{equation}
\mathbb{P}\left(\inf_{x \perp \sigma^* : \Vert x \Vert_2 = 1} x^\top S x > 0 \right) = 1 - o(1). \label{eq:SDP-inf}    
\end{equation}
This is where our proof diverges from \cite{Kim2018}; rather than showing \eqref{eq:SDP-inf}, \cite{Kim2018} proceed through a different sufficient condition.
Using steps similar to the proof of \cite[Theorem 2]{Hajek2016}, we show that for all $x \perp \sigma^*$ such that $\Vert x \Vert_2 = 1$,
\begin{equation}
x^\top Sx \geq \min_{i \in [n]}D_{ii} - \Vert W - W^{*}\Vert_2, \label{eq:Hajek-step}    
\end{equation}
where $W^{*}$ is the expected value of $W$, conditioned on $\sigma^{*}$. It remains to (1) lower-bound $D_{ii}$ for each $i \in [n]$ and (2) upper-bound $\Vert W - W^{*}\Vert_2$.

To lower-bound $D_{ii}$, we condition on $\sigma^{*}$ and establish concentration of $D_{ii}$ around its mean. To see why $D_{ii}$ should be positive and bounded away from zero, it helps to rewrite \eqref{eq:diagonal-matrix} as follows:
\[D_{ii} = \sum_{j \in [n] : \sigma^*(i) = \sigma^*(j)} W_{ij} - \sum_{j \in [n] : \sigma^*(i) \neq \sigma^*(j)} W_{ij}.\]
While the values $\{W_{ij}\}_{j=1}^n$ are dependent, they are functions of \emph{independent} random variables (namely, the hyperedge random variables). After re-expressing $D_{ii}$ in terms of the underlying hyperedge random variables, the proof proceeds by a Chernoff-style argument (Lemma \ref{lem:chernoff bound for hsbm}). Whenever $I(d,\alpha,\beta) >1$, we establish the existence of $\eps > 0$ such that for each $i$, $D_{ii} \geq \eps \log n$ with probability $1 - o(n^{-1})$. A union bound then implies $\min_{i \in [n]} D_{ii} \geq \eps \log n$ with high probability. Next, we need a tail bound on $\Vert W - W^{*}\Vert_2$. While sharp concentration results are known \cite{Lee2020}, we note that we can also bound $\mathbb{E}\left[\Vert W - W^{*}\Vert_2\right]$ using much simpler techniques.
\begin{theorem}[\textbf{Spectral norm expectation}]\label{thm:sp-norm-conc} Let $d \in \{2,3,\dots \}$ be fixed. Let $p \in [0,1]^{\binom{[n]}{d}}$, where $\max_e p_e \leq \nicefrac{c_0 \log n}{\binom{n-1}{d-1}}$ for some constant $c_0>0$. Let $G\sim H(d,n,p)$, and $W=\mathcal{S}(G)$ whose expectation is $W^*$. Then there exists a constant $c\coloneqq c(d,c_0)>0$ such that
$$\Exp{[\twonorm{W-W^*}]} \leq c \sqrt{\log n}.$$
\end{theorem}
Markov's inequality immediately implies the desired tail bound. Returning to \eqref{eq:Hajek-step}, we see that $x^{\top} S x > 0$ simultaneously for all $x$ satisfying $x \perp \sigma^{*}, \Vert x \Vert_2 = 1$, with high probability, concluding \eqref{eq:SDP-inf}.

\paragraph*{Spectral norm concentration.}
We highlight our proof technique for bounding $\mathbb{E}\left[\Vert W - W^{*}\Vert_2\right]$ (Theorem \ref{thm:sp-norm-conc}). Similar bounds are well-known for the spectral norm $\Vert A - \mathbb{E}[A]\Vert_2$ in the case where $A$ is a symmetric matrix with independent, bounded entries and suitably bounded expectation \cite{Feige2005,Lei2015,Hajek2016}.
The first step in our proof is inspired by the symmetrization argument of \cite{Hajek2016}. Let $R$ be a symmetric tensor of order $d$ and dimension $n$ with independent Rademacher entries. Let $G \circ R$ be the hypergraph where each hyperedge is independently assigned to a $+1$ or $-1$ label. Let $\mathcal{S}(G \circ R)$ denote the corresponding similarity matrix, where the $(i,j)$ entry is the sum of $\pm 1$-weighted hyperedges containing $i,j$. We show that
\[\mathbb{E}\left[\Vert W - W^*\Vert_2 \right] \leq 2 \mathbb{E}\left[\Vert\mathcal{S}(G \circ R) \Vert_2\right].\]
A coupling argument then allows us to replace $G$ by $G^{(1)}\sim \text{HSBM}(d, n, p_{\text{max}}, p_{\text{max}})$, where $p_{\max}= \max_e{p_e}$. Unlike the matrix $W - W^{*}$, the matrix $S(G^{(1)} \circ R)$ has entries with identical distributions. However, the entries are dependent.

Our next goal is to create independence. We invoke Jensen's inequality, establishing
\[\mathbb{E}\left[\Vert W - W^*\Vert_2 \right] \leq 2 \mathbb{E}\left[\left \Vert \sum_{m=1}^K \mathcal{S}(G ^{(m)}\circ R^{(m)}) \right\Vert_2\right],\]
where each $G^{(m)}$ is an independent copy of $G^{(1)}$, each $R^{(m)}$ is an independent copy of $R$, and $K = d^2 -d$. Note that $\sum_{m=1}^K \mathcal{S}(G ^{(m)}\circ R^{(m)})$ is a sum of \emph{independent} matrices with \emph{dependent} entries. For any $m \in [K]$, observe that a given hyperedge random variable affects exactly $2 \times \binom{d}{2} = K$ entries of $\mathcal{S}(G ^{(m)}\circ R^{(m)})$. By adding $K$ independent copies, we can rearrange the underlying hyperedge random variables to achieve $\sum_{m=1}^K \mathcal{S}(G ^{(m)}\circ R^{(m)}) = \sum_{k=1}^K C^{(k)},$ where $\sum_{k=1}^K C^{(k)}$ is a sum of \emph{dependent} matrices with \emph{independent} entries. Then \[\mathbb{E}[\Vert W - W^*\Vert_2] \leq 2 \mathbb{E}\left[\left\Vert \sum_{k=1}^K C^{(k)} \right\Vert_2\right] \leq 2 \sum_{k=1}^K \mathbb{E}[\Vert C^{(k)}\Vert_2].\]
The final summation is then straightforwardly bounded using \cite[Theorem 5]{Hajek2016}; see Appendix \ref{appendix:spectral-norm} for a complete proof.

\section{Analysis of the Spectral Algorithm}\label{sec:spectral-outline}
\paragraph*{Correctness of the spectral algorithm.} Recall that $u_2$ denotes the second eigenvector of $W = \mathcal{S}(G)$. Since our algorithm determines communities based on the signs of $u_2$, we need precise bounds on each entry of $u_2$. A natural strategy would be to compare $u_2$ to $u_2^{*}$, since $u_2^* = \frac{1}{\sqrt{n}}\sigma^*$ due to the block structure of $W^*$. Unfortunately, $\Vert u_2 - u_2^*\Vert_{\infty}$ is too large for our purposes, but still $u_2$ recovers communities by sign. To gain intuition for this behavior, write 
\[u_2 - u_2^* = \left(\frac{W u_2^*}{\lambda_2^*}- u_2^{*} \right) + \left(u_2 - \frac{W u_2^*}{\lambda_2^*} \right).\]
\begin{figure}
    \begin{center}
        \includegraphics[scale=0.5]{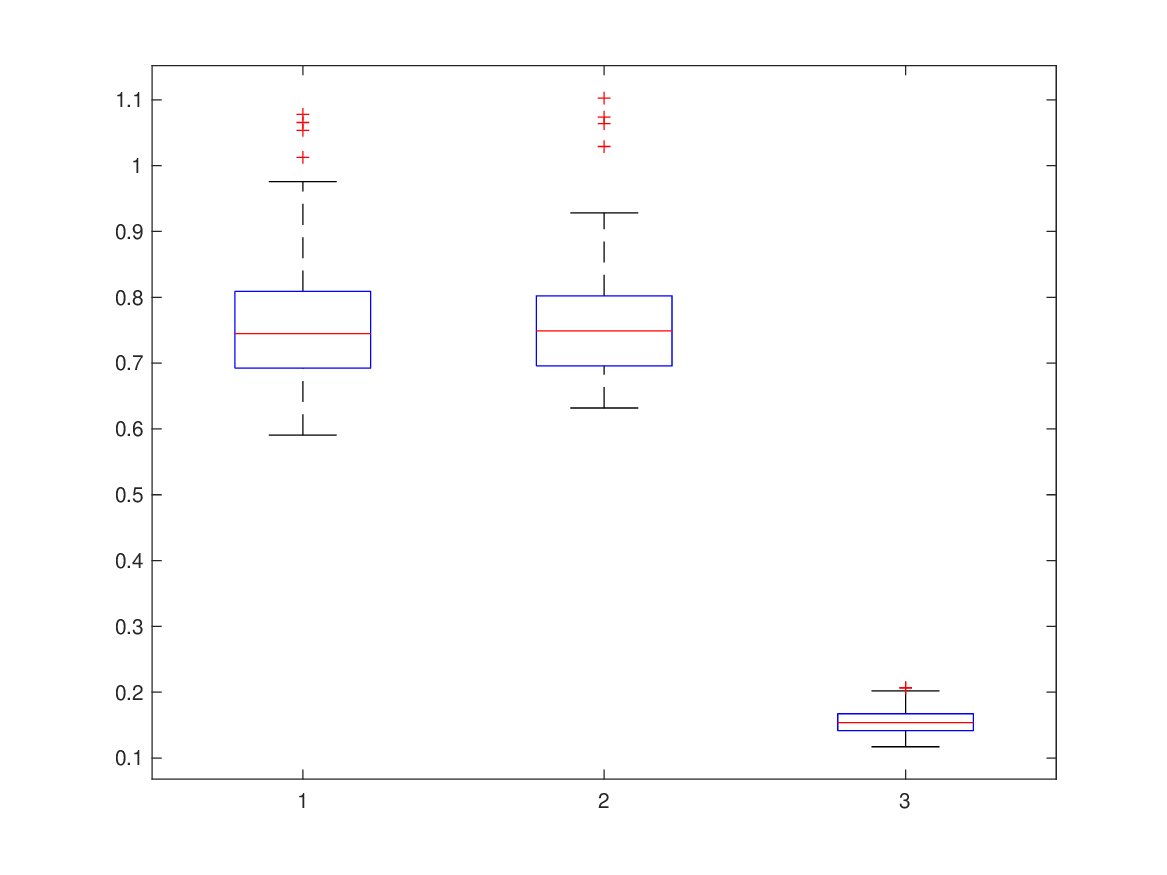}
    \end{center}
    \vspace{-10mm}
    \caption{Consider $ \hsbm(d,n,\alpha f_n,\beta f_n)$ for $d=4$, $n=1000$, $\alpha=50$, $\beta=10$ and $f_n=\log n/\binom{n-1}{d-1}$. The boxplots show three different errors over 100 realizations: (1)$\sqrt{n}\infnorm{ u_2-u_2^*}$, (2)$\sqrt{n}\Vert u_2^*-W u_2^*/\lambda_2^* \Vert_{\infty}$, and (3) $\sqrt{n} \Vert u_2 -W u_2^*/\lambda_2^*\Vert_{\infty} $.}
    \vspace{-5mm}
    \label{fig:boxplots}
\end{figure}
The first term on the right-hand side is the main term, while the second represents a smaller-order term (see Figure \ref{fig:boxplots}). Such behavior was also observed in the SBM setting by \cite{Abbe2020ENTRYWISEEA}.

Our first step is to apply Theorem \ref{thm:entrywise-analysis}, showing that
\begin{equation}
\min_{s^* \in \{\pm 1\}} \left \Vert u_2 - s^*\frac{W u_2^*}{\lambda_2^*}\right \Vert_{\infty} = o(1/\sqrt{n})    \label{eq:asymptotic-entrywise}
\end{equation}
(see Corollary \ref{cor:ewise-for-hsbm}). Therefore, if we can show that the vector $\frac{W u_2^*}{\lambda_2^*}$ has the same signs as $\sigma^*$ (up to a global sign flip), then the same is true for $u_2$. Our goal is then to show that $\sigma^*$ and $Wu_2^*/\lambda_2^*$ have the same signs, i.e. 
$$ \min_{i\in [n]} \sigma^*(i) \left(\frac{W u_2^*}{\lambda_2^*} \right)_i > 0. $$
Fixing the orientation $u_{2}^{*} = \frac{1}{\sqrt{n}} \sigma^*$, we obtain that for $i \in [n]$,
\[\sigma^*(i) \left(\frac{W u_2^*}{\lambda_2^*} \right)_i =  \frac{\sigma^*(i) \sum_{j \in [n]} W_{ij} \sigma^*(j)}{\lambda_2^* \sqrt{n}} =  \frac{D_{ii}}{\lambda_2^* \sqrt{n}}, \]
where $D_{ii}$ is defined in \eqref{eq:diagonal-matrix}. 

In the logarithmic degree regime, we have that $\lambda_2^*=\Theta(\log n)$. Moreover, as in the SDP analysis, we note that for $\epsilon > 0$ sufficiently small, $\min_{i \in [n]} D_{ii} \geq \epsilon \log n$ with high probability, whenever $I(d,\alpha,\beta)>1$. Therefore, the vector $\frac{W u_2^*}{\lambda_2^*}$ has the same signs as $\sigma^*$. Moreover, the entries are of order $1/\sqrt{n}$; in turn, \eqref{eq:asymptotic-entrywise} implies that $u_2$ also has the same signs as $\sigma^*$, up to a global sign flip, implying Theorem \ref{theorem:spectral}. Similarly, when $f_n= \omega(\log n/n^{d-1})$, the we have that $\lambda_2^*=\Theta(n^{d-1} f_n)$. In this case, using a Chernoff-style bound, we establish the existence of $\eps>0$ such that $\min_{i \in [n]} D_{ii} \geq \epsilon \, n^{d-1} f_n$ with high probability for all $\alpha>\beta>0$, giving us Theorem \ref{theorem:spectral-otheregimes}.

\paragraph*{Runtime analysis.} Observe that Step 2 of Algorithm \ref{alg:spectral} requires only $O(n)$ time. The bottleneck is the time required for Step 1 (to compute the second eigenvector of $W$). In order to compute the eigenpairs, one can use the power method on the matrix $W$, which computes the top eigenpair $(\lambda_1,u_1)$ first. The method converges in $O(\log (n)/\delta)$ iterations (see \cite{garber2016faster}), where $\delta = \frac{\lambda_1-\lambda_2}{ \lambda_1}$ is the relative eigengap. In our setting, both $\lambda_1-\lambda_2$ and $\lambda_1$ are $\Theta(n^{d-1} f_n)$. This is because the eigengap $\Delta_1^*=\lambda_1^*-\lambda_2^*$ of $W^*$ is $\Theta(n^{d-1}f_n)$ and $\norm{W-W^*}_2=o(n^{d-1}f_n)$ with high probability using \cite{Lee2020}. We conclude that $\delta=\Theta(1)$ and the power method converges in $O(\log n)$ iterations. In each iteration, it needs to multiply a vector with $W$. The cost of this operation depends on the sparsity of $W$. 

In the logarithmic degree regime, we have $O(n \log n)$ edges in the hypergraph, and thus, $W$ has at most $O(n \log n)$ non-zero entries, which is the effective cost of a matrix-vector multiplication. Therefore, the total time to compute the first eigenpair $(\lambda_1,u_1)$ of $W$ is $O(n \log^2 n)$. To obtain the second eigenvector, we can deflate $W$ by subtracting $\lambda_1  u_1 u_1^\top$. The new relative eigengap is $\frac{\lambda_2-\lambda_3}{\lambda_2}=\Theta(1)$ since the other eigenvalues of $W$ are close to 0.  We emphasize that, even though $W-\lambda_1 u_1 u_1^\top$ may have $n^2$ non-zero entries, a matrix-vector multiplication can still be done in $O(n\log n)$ time by multiplying the vector with $W$ and $\lambda_1 u_1 u_1^\top$ separately and then taking the difference. Thus, the power method requires $O(n \log^2
n)$ time to obtain the second eigenvector.

In super-logarithmic degree regimes, a matrix-vector multiplication may take up to $O(n^2)$ time. Therefore, the total runtime is $O(n^2 \log n)$ in the worst-case.

\paragraph*{Entrwise eigenvector analysis.} \cite{Abbe2020ENTRYWISEEA} introduced a powerful entrywise eigenvector bound, which has been used to show the optimality of spectral algorithms without the need of a clean-up stage \cite{Abbe2020ENTRYWISEEA,Dhara2022a, Dhara2022b, Dhara2022c}. Unfortunately, the entrywise bound \cite[Theorem 2.1]{Abbe2020ENTRYWISEEA} does not apply to $W = \mathcal{S}(G)$, since $W$ violates a certain independence assumption. The independence assumption is critically used in a \emph{leave-one-out} argument; we therefore carefully adapt this step. It also remains to prove a certain \emph{row concentration} property of the matrix $W$.

For simplicity, let $\lambda = \lambda_k$, $\lambda^* = \lambda^*_k$ and $\Delta_k^*=\Delta^*$. For clarity of presentation, we assume the orientation $\langle u, u^{*} \rangle \geq 0$, and make similar simplifying assumptions throughout the outline. Our goal is then to show
$$\infnorm{u -\frac{W{u^*}}{\lambda^*}} \lesssim \hspace{2mm}\frac{ \infnorm{u^*}}{\log \log n}.$$
We first relate $\lambda$ to $\lambda^{*}$. By Weyl's inequality, $|\lambda - \lambda^{*}| \leq \Vert W - W^{*} \Vert_2$. In turn, \cite[Theorem 4]{Lee2020} implies $\Vert W - W^{*}\Vert_2 \leq \gamma \lambda^{*}$ with high probability, for certain $\gamma = \gamma_n = o(1)$. It then follows that  $|\lambda^{-1}-{\lambda^*}^{-1}|  \lesssim \gamma|\lambda^*|^{-1}$. Using this observation along with the triangle inequality, we can show
\begin{align}
\infnorm{u-\frac{Wu^*}{\lambda^*}} &= \infnorm{\frac{Wu}{\lambda} - \frac{Wu^{*}}{\lambda} + \frac{Wu^{*}}{\lambda} -\frac{Wu^*}{\lambda^*}} \hspace{-2mm}\leq \left| \frac{1}{\lambda} -\frac{1}{\lambda^*} \right| \infnorm{Wu^*}+ \frac{1}{|\lambda|} \infnorm{W(u-u^*)} \nonumber \\ 
&\lesssim  \frac{1}{|\lambda^*|}\left(\gamma\infnorm{Wu^*} + \infnorm{W(u-u^*)}\right). \label{eq:entrywise-key-step}
\end{align}
Note that $u^*$ is a deterministic vector. Therefore, in order to bound the term $\Vert W u^{*} \Vert_{\infty}$, we derive a row concentration result (Lemma \ref{lem:row-conc}). For a fixed vector $v \in \mathbb{R}^n$, our row concentration result controls $\Vert W v \Vert_{\infty}$ in terms of both $\Vert v \Vert_{\infty}$ and $\Vert v \Vert_{2}$.

Since $u$ depends on $W$, the second term in \eqref{eq:entrywise-key-step} requires a different strategy than the first. We therefore apply the leave-one-out technique, motivated by other works using a similar strategy \cite{doi:10.1073/pnas.1307845110,doi:10.1073/pnas.1523097113,doi:10.1137/17M1122025,Abbe2020ENTRYWISEEA}. Bounding ~$\Vert W (u -u^{*}) \Vert_{\infty}$ reduces to bounding $\left| \left[W(u - u^{*})\right]_m\right|$ for each $m \in [n]$. To this end, we fix $m \in [n]$ and define a random matrix $W^{(m)}$ which is independent of the $m$-th row and column of $W$. Let $G^{(m)}$ be the hypergraph formed from $G$ by deleting all hyperedges containing $m$, and let $W^{(m)} = \mathcal{S}(G^{(m)})$ be its similarity matrix. Let $\um{m}$ be the $k$-th eigenvector of $W^{(m)}$. 
Applying the triangle and Cauchy--Schwarz inequalities, we obtain
  \begin{align}
      \left|\left[W(u-u^*)\right]_m \right|  & \leq \left|W_{m\cdot}(u -\um{m}) \right| + \left| W_{m\cdot} (\um{m}-u^*) \right| \nonumber\\
      & \leq \twotoinfnorm{W}\Vert u -\um{m} \Vert_2 + | W_{m\cdot} (\um{m}-u^*) |. \label{eq:outline-1}
  \end{align}
 Observe that $W_{m\cdot}$ and $\um{m}-u^*$ are independent by the leave-one-out construction. Therefore, we can bound the second term in \eqref{eq:outline-1} using our row concentration result. In order to bound $\Vert u -\um{m} \Vert_2$, we apply a version of the Davis--Kahan $\sin(\theta)$ theorem \cite{Davis1970TheRO}, which yields
 \[\Vert u -\um{m} \Vert_2 \lesssim \frac{\Vert (W - W^{(m)}) u\Vert_2}{\Delta^*}.\]

Our analysis so far is almost identical to that of \cite{Abbe2020ENTRYWISEEA}; the main difference arises here. Note that the $m$-th row and column $W-\A{m}$ are the same as those of $W$. Since a hyperedge containing the vertex $m$ contributes to other entries of $W$, there will be additional non-zero entries outside of the $m$-th row and column. Thus it is harder to find tight bounds on the quantities that essentially are of main interest: $\Vert W-W^{(m)}\Vert_2$ and $\Vert (W-W^{(m)})u\Vert_2$. This complication arises whenever $d \geq 3$, and is absent from the analysis of \cite{Abbe2020ENTRYWISEEA}, due to their independence assumption. However, we are still able to prove similar probabilistic bounds for these quantities via a series of careful (and non-trivial) lemmas that uses structural properties of the similarity matrix along with the spectral norm concentration of \cite{Lee2020} to yield the final sharp bounds. 
 
 \section{Discussion and future work}\label{sec:discussion}
This paper contributes to a line of work which studies random graph inference problems under restricted information. This paper considers \emph{aggregated} information through the similarity matrix transformation. Other recent works consider \emph{noisy} or \emph{censored} \cite{Abbe2014,Hajek2016b,Dhara2022a,Dhara2022b,Dhara2022c}] information models. While the similarity matrix is a lossy representation of the full hypergraph information, it retains most of the information about the \emph{latent} community structure. That is, the similarity matrix is sufficient for exact recovery in all denser regimes, and even in the logarithmic degree regime, at least up to the threshold $I(d,\alpha,\beta)=1$ which is slightly worse than the threshold given the full information (see Figure \ref{fig:phase-transition}). In Appendix \ref{subapp:sim vs adj}, we further investigate the recovery problem given the adjacency matrix $A \in \{0,1\}^{n \times n}$, where $A_{ij} = 1$ if $i,j$ appear together in some hyperedge. In the logarithmic degree regime, our results suggest that the adjacency matrix preserves much of the latent community structure, while in higher degree regimes the adjacency matrix becomes uninformative.

Our work shows that, in the logarithmic degree regime, the spectral algorithm matches the performance of min-bisection estimator, i.e. it succeeds whenever $I(d,\alpha,\beta)>1$. It is not clear whether the exact recovery given only the similarity matrix is possible when $I(d,\alpha,\beta)<1$ but $I_\text{\small full}(d,\alpha,\beta)>1$. Deriving a precise information-theoretic threshold for exact recovery given the similarity matrix remains an important direction for future work. We believe that $I(d,\alpha,\beta)$ is also the information-theoretic limit, which would make the spectral and the SDP algorithms optimal.  
 
Finally, the spectral algorithm is advantageous even when the full information is given, due to its computational efficiency compared to existing approaches. To our knowledge, this constitutes the first algorithm with a nearly linear runtime for exact recovery in the HSBM.

We include some directions for future work.
\begin{enumerate}
    \item What is the sharp information-theoretic threshold for exact recovery given only the similarity matrix in the logarithmic degree regime?
    \item We show that the SDP is robust to a monotone adversary who operates on the similarity matrix. What can we say about an adversary who instead operates on the hypergraph directly, by adding intra-community edges and removing inter-community edges? 
    \item Can we find a simple proof for the sharp concentration $\Vert W - W^*\Vert_2$, by leveraging our bound on $\mathbb{E}\left[\Vert W - W^*\Vert_2\right]$ and using a concentration argument? While Talagrand's inequality leads to a sharp bound in the $d = 2$ case \cite{Hajek2016}, it becomes vacuous in the $d \geq 3$ case.
    \item What is the precise recovery threshold given the adjacency matrix $A$, in the logarithmic degree regime? Does the same spectral strategy (of recovering communities based on the signs of the entries of the second eigenvector of $A$) achieve the information-theoretic limit? 
\end{enumerate}

\paragraph*{Acknowledgements} J.G. was partially supported by NSF CCF-2154100. N.J. was partially supported by NSF CAREER Award 1652491. The authors thank the anonymous reviewers for their helpful feedback. N.J. would like to deeply thank Ludovic Stephan for an insightful discussion on the MLE and the min-bisection estimator.

\bibliographystyle{alpha}
\bibliography{general}

\newpage
\appendix

 \section{Similarity Matrix vs. Adjacency Matrix Information Models}\label{subapp:sim vs adj}
In this section, we study the exact recovery problem given only the adjacency matrix $A \in \{0,1\}^{n \times n}$, where $A_{ij} = 1$ whenever $i,j \in [n]$ appear together in some hyperedge in $G$. For example, an academic collaboration network can be modeled as a hypergraph, where each hyperedge represents the authors of a paper. However, academic collaboration networks typically only record pairwise information, describing whether two researchers have co-authored a paper together--motivating the need to study this problem. 

\paragraph*{\textbf{Logarithmic degree regime.}}We first consider $f_n = \log n/\binom{n-1}{d-1}$. We conjecture that the exact recovery problem given $A$ exhibits a sharp recovery threshold, analogously to the exact recovery problem given $W$. Additionally, we expect the spectral algorithm that recovers communities based on the second eigenvector of $A$ to be optimal. In Figure \ref{fig:sim vs adj_plot}, we present our empirical findings from simulations.

\begin{figure}[ht]
    \centering
    \includegraphics[width=90mm, height=80mm]{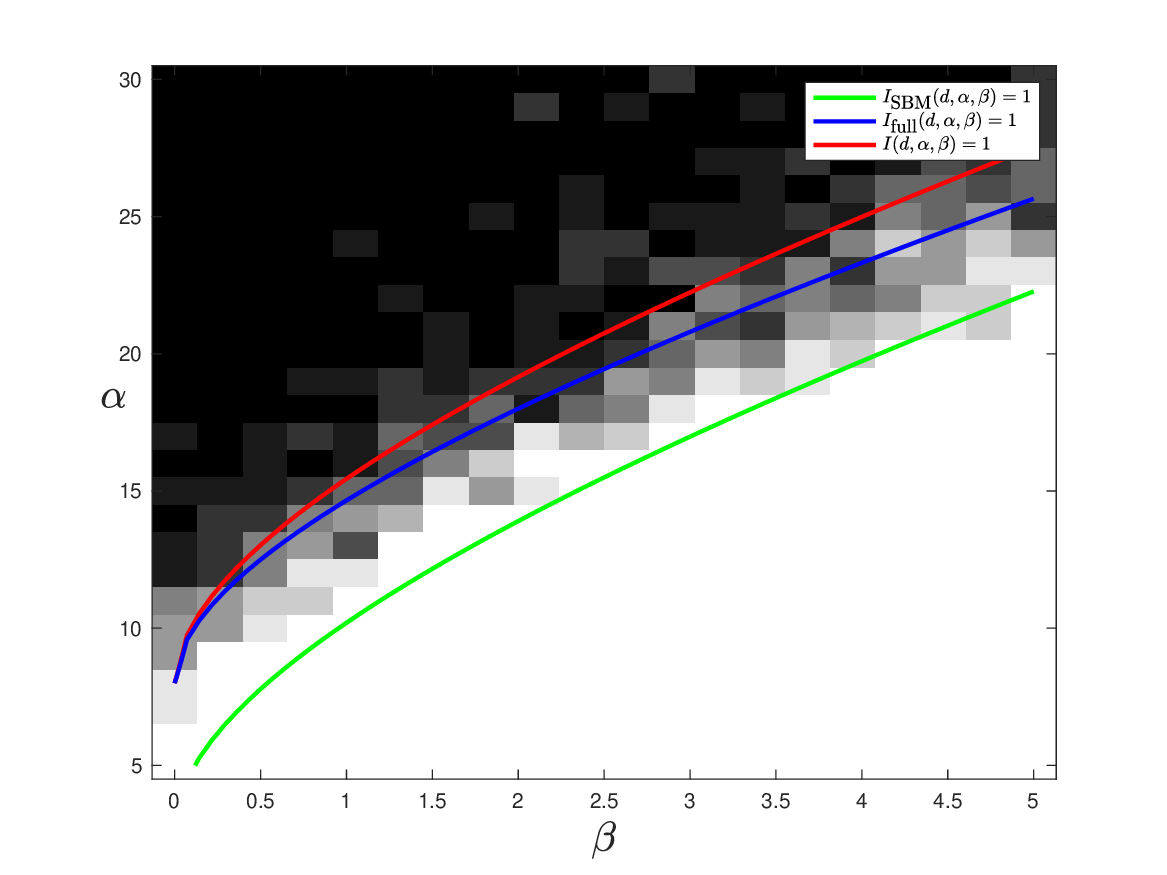}
    \caption{Visualizing the heat map of success of spectral algorithm on $A$ alongside $I(d,\alpha,\beta)>1$.}
    \label{fig:sim vs adj_plot}
\end{figure}

Fix $d=4$ and $n=500$. Let $G \sim \hsbm(d, n, \alpha f_n, \beta f_n)$ for different values of $(\alpha, \beta)$ and $W=\cS(G)$. Then, the adjacent matrix $A = \min\{W, \mathbf{11}^\top\}$. We run the spectral algorithm for various values $\alpha>\beta$. We report the proportion of success, namely $\hat{\sigma} \in \{ \pm \sigma^*\}$, out of 30 independent runs. Darker pixels represent higher chances of success in the heat map. We juxtapose this with the thresholds $I(d,\alpha,\beta)=1$ and $I_{\textnormal{full}}(d,\alpha,\beta)=1$. We find all three of them to be relatively close to each other, which suggests that the matrix $A$ still continues to retain much of the information about the ground truth community structure.

Since $A$ is an adjacency matrix, it is also natural to compare it to an SBM $G'$ which produces an adjacency matrix $A'$ with the same marginal edge probabilities. Of course, $A$ has a complicated dependency structure, while the entries of $A'$ are independent, conditioned on the community structure. We plot the recovery threshold for the SBM in green. More precisely, we transform the parameters $(\alpha, \beta)$ to corresponding SBM parameters $(\alpha',\beta')$ by solving the following. The intra-community probability 
$$p_n'=\frac{\alpha' \log n}{n} = 1 - \left(1 - \alpha f_n \right)^{\binom{n/2-2}{d-2}} \left(1 - \beta f_n \right)^{\binom{n-2}{d-2} - \binom{n/2-2}{d-2}},$$ 
and inter-community probability $$q_n'=\frac{\beta' \log n}{n} = 1 - \left(1 - \beta f_n \right)^{\binom{n-2}{d-2}}.$$
This gives us $\alpha'$ and $\beta'$ as functions of $\alpha,\beta$ and $d$. The line plots $I_{\textnormal{SBM}}(d,\alpha,\beta)=(\sqrt{\alpha'} - \sqrt{\beta'})^2/2 = 1$ on the $(\alpha, \beta)$ plane. This line falls below the recovery threshold corresponding to the full hypergraph information, which implies that the threshold corresponding to recovery from $A$ is strictly higher than the threshold corresponding to recovery from $A'$. In other words, even though $A$ and $A'$ have the same marginal entry distributions, the dependency structure in $A$ makes recovery strictly harder.

From a memory and runtime standpoint, recovery using $A$ or $W$ is asymptotically equivalent in the logarithmic degree regime, since with high probability each entry of $W$ is upper bounded by $4d = \Theta(1)$ (see Section \ref{appendix:spectral-norm}, Equation \ref{eq:F-c}). On the other hand, there are regimes in which the similarity matrix approach succeeds while the adjacency matrix approach fails, which we discuss below. 

\paragraph*{\textbf{Denser regimes.}} Fix $d$ and $\alpha>\beta>0$. Consider regimes such that $f_n=\omega(\log n/ n^{d-2})$ and $G\sim (d,n,\alpha f_n,\beta f_n)$. For each pair $i,j \in [n]$, the number of edges in involving $(i,j)$ grows with $n$. Therefore, with high probability, the adjacency matrix $A$ is simply the matrix of all ones (up to the diagonal). More precisely,  $\prob{A= (\mathbf{11}^\top -\mathbf{I}) \,} = 1-o(1)$, rendering the adjacency matrix uninformative.

In these regimes, the adjacency matrix information model (or its variant where $A=\min\{W,c\mathbf{11}^\top\}$ for any constant $c>0$) requires only up to $O(n^2)$ space to maintain, but it does not preserve any information. On the other hand, community recovery given the similarity matrix is still possible (Theorem \ref{theorem:spectral-otheregimes}).

\section{Proofs from Section \ref{sec:SDP-outline}: Analysis of the SDP Relaxation}\label{app:SDP-proofs}
We begin by noting a simple corollary of our spectral norm concentration theorem (Theorem \ref{thm:sp-norm-conc}) in the context of the HSBM.
\begin{corollary}\label{cor:hsbm-sp-norm-log-regime}
    Fix $d\in\{2,3,\dots\}$ and $\alpha>\beta>0$. Let $f_n=\log n/\binom{n-1}{d-1}$ and $G \sim \hsbm(d,n,\alpha f_n, \beta f_n)$. Let $W=\mathcal{S}(G)$ and $W^*=\Exp[W \mid \sigma^*]$. Then
    $$\prob{\norm{W-W^*}_2 \leq \log^{3/4} n} \geq 1-o(1).$$
\end{corollary}
\begin{proof}
   By Theorem \ref{thm:sp-norm-conc}, $\Exp[\norm{W-W^*}_2] \leq c \sqrt{\log n}$, for some $c>0$ that depends on $\alpha$ and $d$. Therefore, using Markov's inequality,
   $$\prob{\norm{W-W^*}_2 > \log^{3/4} n } \leq \, \frac{\Exp[\norm{W-W^*}_2]}{\log^{3/4}n} \, \leq \,\frac{c\sqrt{\log n}}{\log^{3/4}n}= \, o(1).\qed$$
   
\end{proof}
We now show Lemma \ref{lem:uniqueness of sdp solution}, which characterizes sufficient conditions under which $X^*=\sigma^* {\sigma^*}^{\top}$ is the unique solution of the SDP \eqref{eq:SDP},  using a dual certificate strategy. 
\begin{proof}[Proof of Lemma \ref{lem:uniqueness of sdp solution}]
We first show that $X^*=\sigma^* {\sigma^*}^\top$ is an optimal solution. For any $X$ satisfying the constraints in \eqref{eq:SDP},
\begin{align*}
    \ip{W}{X}& \leq \ip{W}{X}+\ip{S}{X} \tag{since $S,X \succeq 0$ }\\
   & = \ip{W}{X}+\ip{D+\nu \mathbf{11}^\top -W}{X} = \ip{D}{X} \tag{since $X$ is primal feasible}\\
   &= \ip{D}{X^*} \tag{as $X^*_{ii}=X_{ii}=1$ for all $i \in [n]$}\\
   &= \ip{W+S-\nu \mathbf{11}^\top}{X^*} = \ip{W}{X^*}+\ip{S}{X^*} = \ip{W}{X^*}. \tag{since $ \ip{S}{X^*}= (\sigma^*)^\top S\sigma^*=0$}
\end{align*}
Therefore it only remains to establish the uniqueness of $X^*$. Consider an optimal solution $\Tilde{X}$. Then
\begin{align*}
    \ip{S}{\Tilde{X}}&= \ip{D+\nu \mathbf{11}^\top-W}{\Tilde{X}} = \ip{D-W}{\Tilde{X}} \tag{since $\ip{\mathbf{11}^\top}{\Tilde{X}}=0$ }\\
    & = \ip{D-W}{X^*} \tag{as $\ip{W}{X^*}=\ip{W}{\Tilde{X}}$ and $X^*_{ii}=\Tilde{X}_{ii}=1$}\\
    &= \ip{S}{X^*}= 0. \tag{using $S\sigma^*=0$}
\end{align*}
Since $\Tilde{X} \succeq 0$ and $S \succeq 0$ and $\lambda_{n-1}(S)>0$, we obtain that $\Tilde{X}$ is also a rank-1 matrix and hence it must be a multiple of $\sigma^* {\sigma^*}^\top$. Moreover, as $\Tilde{X}_{ii}=1$ for all $i \in [n]$, it must be that $\Tilde{X}=X^*=\sigma^* {\sigma^*}^\top$.
\end{proof}

We now focus on showing that the choice of $D$ mentioned in \eqref{eq:diagonal-matrix} and $\nu=1$ satisfy the conditions in Lemma \ref{lem:uniqueness of sdp solution} with high probability whenever $I(d,\alpha,\beta)>1$ to show Theorem \ref{theorem:SDP}. Towards this, we prove a lemma that plays an important role in proving the theorem. Roughly speaking, it provides a probabilistic lower bound on $D_{ii}$ defined in \eqref{eq:diagonal-matrix} for any $i\in[n]$. 
\begin{lemma}\label{lem:chernoff bound for hsbm} 
Let $d \in \{2, 3, \dots\}$, and $\alpha > \beta > 0$, such that $I(d, \alpha, \beta) > 1$. Let $f_n=\log n /\binom{n-1}{d-1}$ and $W=\mathcal{S}(G)$ where $G\sim \hsbm(d,n,\alpha f_n,\beta f_n)$. Then there exists a constant $\eps:=\eps(d,\alpha,\beta)>0$ such that for any fixed $i\in [n]$, with probability at least $1-o(n^{-1})$,
$$ \sum_{j\in [n]} W_{ij} \sigma^*(i) \sigma^*(j) \geq \eps \log n.$$
\end{lemma}

\begin{proof}
Fix $i \in [n]$, and let $X \triangleq  \sum_{j\in [n]} W_{ij} \sigma^*(i) \sigma^*(j)$. Let $\mathcal{E}\coloneqq \binom{[n]}{d}$ be the set of possible edges. For each $e \in \mathcal{E}$, let $A_e$ be the indicator that edge $e$ is present. Let $\Em{i}:=\{ e\in \mathcal{E} : i \in e \}$ represent the set of potential edges incident on $i$.

For any edge $e\in \Em{i} $, let $n_i(e):=| \{ j\in e \setminus \{i\}: \sigma^*(i) \neq \sigma^*(j) \} |$ be the number of vertices in $e$ that belong to the opposite community as $i$. We can rewrite $X$ as follows
\begin{align*}
X &= \sum_{e \in \mathcal{E}^{(i)}}  \sum_{j \in e, j \neq i} \sigma^*(i) \sigma^*(j) A_e\\
&= \sum_{e \in \mathcal{E}^{(i)}}\left((d-1 - n_i(e))-n_i(e) \right) A_e\\
&= \sum_{e \in \mathcal{E}^{(i)}}\left( d - 1- 2n_i(e)  \right) A_e.
\end{align*}
Next, observe that for $r \in \{0, 1, \dots, d-1\}$, the set $\{e \in \mathcal{E}^{(i)} : n_i(e) = r\}$ has cardinality 
\[N_r :=  \binom{n/2}{r} \binom{n/2-1}{d-1 -r}.\]
Let $\{Y_r\}_{r=0}^{d-1}$ be independent random variables, where $Y_r \sim \text{Bin}(N_r, q_r)$, with 
\[
q_r = \begin{cases}
      \nicefrac{\alpha \log n}{\binom{n-1}{d-1}}, & \text {if } r=0;\\
      \nicefrac{\beta \log n}{ \binom{n-1}{d-1}}, & \text{if } 1 \leq r \leq d-1.
\end{cases}
\]
We then further rewrite $X$ as follows:
\begin{align*}
X &= \sum_{r=0}^{d-1} \sum_{e \in \mathcal{E}^{(i)}} \mathbbm{1}\{n_i(e) = r\}\left(d - 1 - 2n_i(e) \right) A_e,
\end{align*}
so that $X$ is equal to $\sum_{r=0}^{d-1} (d-1-2r) Y_r$ in distribution.

Let $h_r = d-1-2r$. Fix $\epsilon \in \mathbb{R}$ and $t \geq 0$. Exponentiating and applying Markov's inequality,
\begin{align}
    \prob{X \leq \eps \log n}  &\leq \prob{e^{-tX} \geq e^{-t \eps \log n}} \nonumber \\
    & \leq \frac{\Exp{[e^{-tX}]}}{e^{-t\eps \log n}} \nonumber \\
    &= e^{t\eps \log n} \Exp{[e^{-t\sum_{r=0}^{d-1} h_r Y_r}]} \nonumber  \\
    &= e^{t\eps \log n} \prod_{r=0}^{d-1} \Exp{[e^{-t h_r Y_r}]}  \nonumber \\
    & =  e^{t \eps \log n} \prod_{r=0}^{d-1} \left(1-q_r (1- e^{-t h_r} )\right)^{N_r} \nonumber \\
   & \leq \text{exp} \left(t \eps \log n - \sum_{r=0}^{d-1} N_r q_r (1-e^{-h_rt})  \right) \label{eq:Chernoff}.
 \end{align}
Here, the second equality is due to independence of the $Y_r$ random variables, and the final step uses $1-x \leq e^{-x}$. Next,
\begin{align*}
N_r &= \binom{n/2}{r} \binom{n/2-1}{d-1 -r}=\frac{\binom{n/2}{r} \binom{n/2-1}{d-1 -r}}{\binom{n-1}{d-1}} \binom{n-1}{d-1}\\
&=(1+o(1)) \frac{(n/2)^r}{r\,!} \cdot \frac{(n/2)^{d-1-r}}{(d-1-r)\,!} \cdot \frac{(d-1)!}{n^{d-1}} \cdot \binom{n-1}{d-1}\\
&=(1+o(1)) \frac{1}{2^{d-1}} \binom{d-1}{r} \binom{n-1}{d-1}.
\end{align*}
Substituting into \eqref{eq:Chernoff}, we obtain
\begin{align*}
\mathbb{P}\left(X \leq \epsilon \log n\right) &\leq \exp \left( t \eps \log n - \frac{(1+o(1))}{2^{d-1}} \bigg(  \alpha (1-e^{-(d-1)t}) +\sum_{r=1}^{d-1}  \beta \binom{d-1}{r}(1-e^{-(d-1-2r)t})\bigg)\log n\right).
\end{align*}
Let $t = t^*(d,\alpha,\beta)$, where
$$t^*(d, \alpha, \beta)=\arg \max_{t\geq 0} \hspace{2mm} \frac{1}{2^{d-1}}\left(  \alpha (1-e^{-(d-1)t}) +\sum_{r=1}^{d-1}  \beta \binom{d-1}{r} (1-e^{-(d-1-2r)t})\right) := \arg \max_{t \geq 0} \psi(t).$$
We then obtain
\begin{align*}
   \prob{X \leq \eps \log n} & \leq \text{exp} \big(t^* \eps \log n - I(d,\alpha,\beta) \log n + o(\log n)  \big) \\
   &\leq n^{-I(d,\alpha,\beta)+ t^* \eps + o(1)}.
\end{align*}
Note that $t^* \neq 0$ as 
\[\lim_{t \rightarrow 0^{+}} \psi^{\prime} (t) =  \frac{1}{2^{d-1}}\left( \alpha(d-1) + \beta \sum_{r=1}^{d-1} \binom{d-1}{r}(d-1-2r) \right) = (\alpha-\beta)(d-1)/2^{d-1}>0.\]Furthermore, since $I(d,\alpha,\beta)>1$, one can choose $\eps=\eps(d,\alpha,\beta)>0$ sufficiently small such that
$$\prob{ \sum_{j\in [n]} W_{ij} \sigma^*(i) \sigma^*(j) \leq \eps \log n} = o(n^{-1}). \qed$$

\end{proof}
Finally, we make some important observations about the structure of $W^*$. Observe that $W^*$ has a \emph{block} structure (up to the diagonal entries). In particular, $W^*$ is a zero diagonal symmetric matrix whose non-diagonal entries are given by 
\begin{equation}
 W^*_{ij} =
\begin{cases}
    p' \triangleq \binom{n/2-2}{d-2} \alpha f_n +  \left( \binom{n-2}{d-2}-\binom{n/2-2}{d-2} \right) \beta f_n, & \text{if } \sigma^*(i)=\sigma^*(j);\\ 
    q' \triangleq \binom{n-2}{d-2} \beta f_n,             & \text{if } \sigma^*(i) \neq \sigma^*(j).
\end{cases}   \label{eq:block-structure}
\end{equation}
Observe that $W^*$ can be decomposed as
\begin{equation}\label{eq:W* decomposition}
    W^*=\left( \frac{p'+q'}{2} \right) \mathbf{11}^\top + \left( \frac{p'-q'}{2} \right) \sigma^* {\sigma^*}^\top - p' \mathbf{I}.
\end{equation}

\begin{proof}[Proof of Theorem \ref{theorem:SDP}]
The proof uses ideas from the proof of \cite[Theorem 2]{Hajek2016}.
Let $\nu =1$ and $S \triangleq D + \nu \mathbf{11}^\top - W = D + \mathbf{11}^\top - W$. The goal is to show that $S$ satisfies the conditions mentioned in Lemma \ref{lem:uniqueness of sdp solution} with high probability whenever $I(d,\alpha,\beta)>1$. Observe that, by definition of $D$ in \eqref{eq:diagonal-matrix}, for any $i\in[n]$, we have $D_{ii} \sigma^{*}(i) =\sum_{j\in [n]} W_{ij} \sigma^*(j)$; i.e. $D\sigma^*=W \sigma^*$. Therefore, using the fact that $\ip{\mathbf{1}}{\sigma^*}=0$, we get
$$S\sigma^*=  D\sigma^* + \mathbf{11}^\top \sigma^* -W\sigma^* = 0.$$

Therefore, it remains to show that
$$\prob{ \left\{\inf_{x \perp \sigma^{*} : \Vert x \Vert_2 = 1} x^\top S x > 0 \right\}} \geq 1-o(1).$$
For any $x \perp \sigma^*$ such that $\twonorm{x}=1$, 
\begin{align*}
    x^\top S x &= x^\top D x+  x^\top \mathbf{11}^\top x - x^\top (W-W^*) x - x^\top W^* x\\
    &= x^\top D x +  (\mathbf{1}^\top x)^2 - x^\top (W-W^*) x - \left(\frac{p'-q'}{2} \right)(x^\top \sigma^*)^2  - \left( \frac{p'+q'}{2} \right) (\mathbf{1}^\top x)^2 + p' \tag{using \eqref{eq:W* decomposition} to substitute $W^*$}\\
    &=x^\top D x + \left(1- (p'+q')/2 \right) (\mathbf{1}^\top x)^2 - x^\top (W-W^*) x + p' \tag{since $x \perp \sigma^*$}\\
        & \geq x^\top D x + p' - x^\top (W-W^*) x   \tag{since $p',q'=\Theta(\log n /n)$ are vanishing}\\
        & \geq \min_{i\in [n]} D_{ii} - \twonorm{W-W^*}.   \tag{by the definition of $\twonorm{.}$ for matrices and the fact $p' \geq 0$}
\end{align*}
We now use Lemma \ref{lem:chernoff bound for hsbm} and take a union bound over $i$ to obtain $\min_{i\in [n]} D_{ii} \geq \eps \log n$ with probability $1-o(1)$. Moreover, applying Corollary \ref{cor:hsbm-sp-norm-log-regime}, $\twonorm{W-W^*}\leq \log^{3/4} n$  with probability $1-o(1)$. Therefore, one can conclude that $x^\top S x \geq \eps \log n -\log^{3/4} n>0$ for any $x$ such that $\twonorm{x}=1$ and $x\perp \sigma^*$, completing the proof.
\end{proof}
We additionally show that the SDP is robust to a monotone adversary (Lemma \ref{lemma:monotone-adversary}). Here, we consider an adversary who can increase the value of $W_{ij}$ for any $(i,j)$ in the same community, and decrease the value of $W_{ij}$ for any $(i,j)$ in opposite communities. The robustness of SDPs to monotone adversaries is well-known (see e.g. \cite{Feige2000, Feige2001}). While the monotone adversary appears to provide helpful information, spectral algorithms generally fail under such a semirandom model.

\begin{proof}[Proof of Lemma \ref{lemma:monotone-adversary}]
Let $X$ be a feasible solution to the modified SDP, and let $X^*$ be the unique optimal solution to \eqref{eq:SDP}, which is guaranteed by Theorem \ref{theorem:SDP}. Due to uniqueness of $X^{*}$, we have $\langle W, X \rangle < \langle W, X^*\rangle$. Since $X\succeq 0$, we can write 
\begin{align*}
X = \sum_{l=1}^n \lambda_l v_l v_l^\top
\end{align*}
as its eigendecomposition, where $\lambda_l \geq 0$ for all $l$. Then by the Cauchy--Schwarz inequality,
\begin{align*}
X_{ij}^2 &= \left(\sum_{l=1}^n \lambda_l v_{l,i} v_{l,j}\right)^2\\
&\leq \left(\sum_{l=1}^n \lambda_l v_{l,i}^2 \right) \left(\sum_{l=1}^n \lambda_l v_{l,j}^2 \right)\\
&= X_{ii} \cdot X_{jj}\\
&= 1.
\end{align*}
Therefore, $|X_{ij}| \leq 1$ for all $i,j$, which implies \[\ip{\widetilde{W}-W}{X} \leq \sum_{i,j \in [n]} |\widetilde{W}_{ij}-W_{ij}| = \ip{\widetilde{W}-W}{X^*}.\] Consequently, $$\ip{\widetilde{W}}{X}= \ip{W}{X}+ \ip{\widetilde{W}-W}{X} < \ip{W}{X^*}+ \ip{\widetilde{W}-W}{X^*} = \ip{\widetilde{W}}{X^*},$$ establishing the unique optimality of $X^*$.
\end{proof}
One could also consider another natural monotone adversary that operates directly on the underlying hypergraph instead of on its similarity matrix. More specifically, the adversary is allowed to add intra-community edges and delete cross-community edges. It is not immediately clear whether an SDP algorithm would be robust to such an adversary.

\section{Proofs from Section \ref{sec:spectral-outline}: Analysis of the Spectral Algorithm}\label{sec:spectral}
In Section \ref{subapp:spectral analysis log regime}, we prove that the spectral algorithm succeeds up to the min-bisection threshold in the logarithmic degree regime (Theorem \ref{theorem:spectral}). In Section \ref{subapp:spectral recovery in dense regime}, we establish its correctness in denser regimes (Theorem \ref{theorem:spectral-otheregimes}). Finally, in Section \ref{subapp:ewise analysis}, we present the proof of our general entrywise eigenvector bound (Theorem \ref{thm:entrywise-analysis}).   

\subsection{Proving Theorem \ref{theorem:spectral}}\label{subapp:spectral analysis log regime}
We first derive a corollary of Theorem \ref{thm:entrywise-analysis}, specific to HSBMs.
\begin{corollary}\label{cor:ewise-for-hsbm}
Fix $d\in\{2,3,\dots\}$. Choose any $\alpha>\beta>0$ and $f_n$ according to \eqref{eq:parameter-regimes}. Let $G\sim \textnormal{HSBM}(d,n,\alpha f_n,\beta f_n)$ and $W=\cS(G)$. If $f_n=\Omega(\log n/n^{d-1})$, then with probability at least $1-O(n^{-3})$,
$$ \min_{s^* \in \{\pm 1\}} \infnorm{ u_2 - s^* \frac{W u_2^*}{\lambda_2^*}} \hspace{-1mm}\leq \hspace{1mm} \frac{c}{\sqrt{n}\log \log n},$$
where $c:=c(d,\alpha,\beta)$ is some positive constant that only depends on $d,\alpha$, and $\beta$.
\end{corollary}
\begin{proof}
Since each hyperedge exists with probability at most $\alpha f_n$, we can set $\mu_n$ as $\alpha f_n$. Moreover, we have that $n\binom{n-1}{d-1} \mu_n \geq c_0 \log n$ when $f_n=\Omega(\log n/n^{d-1})$, for some $c_0$ that depends on $d$ and $\alpha$. We now verify that Assumption \ref{asump:well-separated} holds. Let $W^*$ be the expectation of $W$ conditioned on $\sigma^*$, whose entries are $p'$ and $q'$ as given by \eqref{eq:block-structure}. By \eqref{eq:W* decomposition}, $W^*+p'\mathbf{I}$ is a rank-2 matrix. Its non-zero eigenvalues are $(p'+q')n/2$ and $(p'-q')n/2.$ Accounting for the diagonal matrix $p' I$, the eigenvalues of $W^{*}$ are given by
\[\lambda_1^* = (1+o(1))(p'+q')n/2 \text { and } \lambda_2^* = (1+o(1))(p'-q')n/2.\]
Since $p' \asymp q'\asymp n^{d-2} f_n$, we have $\lambda_1^*, \lambda_2^* \asymp n^{d-1} f_n$. Furthermore,
$$\Delta_2^*=\min \left\{ \lambda_1^*-\lambda_2^*, \lambda_2^*-0\right\}= (1+o(1)) \min\left\{ q'n, \frac{(p'-q')n}{2}\right \} \asymp n^{d-1} f_n,$$ hiding constants in $\alpha,\beta$ and $d$. 
Therefore, Theorem \ref{thm:entrywise-analysis} applies.

It remains to verify that $\infnorm{u_2^*}= O(1/\sqrt{n})$. We already know that the second eigenvector of $W^*+p'\mathbf{I}$ is $\frac{1}{\sqrt{n}}\sigma^*$ by \eqref{eq:W* decomposition}. Hence  
\begin{equation}\label{eq:sigma^* entries}
   u_2^*=\frac{1+o(1)}{\sqrt{n}} \hspace{1mm}\sigma^*. \qed
\end{equation}

\end{proof}
We next prove the correctness of the spectral algorithm in the logarithmic degree regime (Theorem \ref{theorem:spectral}).
\begin{proof}[Proof of Theorem \ref{theorem:spectral}]
Recall that $f_n=\log n/\binom{n-1}{d-1}$. Let us fix $s=s^*$ for which Corollary \ref{cor:ewise-for-hsbm} holds. Using the corollary, with probability $1-o(1)$,
\begin{equation}\label{eq:signs match}
   \sqrt{n} \min_{i \in [n]} s \sigma^*(i) u_{2,i} \geq   \sqrt{n} \min_{i \in [n]} s^2 \sigma^*(i) (W u_2^*)_i/\lambda_2^* - c(\log \log n)^{-1},  
\end{equation}
where $c$ is defined in Corollary \ref{cor:ewise-for-hsbm}. Note that $s^2=1$. Also, using \eqref{eq:sigma^* entries}, 
\begin{align*}
    \sqrt{n}\sigma^*(i)(Wu_2^*)_i = (1+o(1))  \sum_{j\in [n]} W_{ij} \sigma^*(i) \sigma^*(j). 
\end{align*}
By Lemma \ref{lem:chernoff bound for hsbm}, if $I(d,\alpha,\beta)>1$, then there exists a positive constant $\eps(d,\alpha,\beta)>0$ such that for a fixed $i \in [n]$, $ \sum_{j\in [n]} W_{ij} \sigma^*(i) \sigma^*(j) \geq \eps \log n$ with probability $1-o(n^{-1})$. Therefore, a union bound implies that with probability $1-o(1)$, 
\[\sqrt{n} \min_{i \in [n]} s^2 \sigma^*(i) (W u_2^*)_i \geq (1+o(1))\epsilon \log n.\]
Since $\lambda_2^* \asymp \log(n)$ when $f_n=\log n/\binom{n-1}{d-1}$, \eqref{eq:signs match} implies that there exists $\eta > 0$ such that
\begin{align*}
\sqrt{n} \min_{i \in [n]} s^* \sigma^*(i) (u_2)_i &\geq   (1+o(1))\eps \log n/\lambda_2^* - c(\log \log n)^{-1}> \eta,
\end{align*}
with probability $1-o(1)$, concluding the proof.
\end{proof}

\subsection{Proving Theorem \ref{theorem:spectral-otheregimes}}\label{subapp:spectral recovery in dense regime}
In this subsection, we prove the correctness of the spectral algorithm (Algorithm \ref{alg:spectral}) in super-logarithmic degree regimes. By Corollary \ref{cor:ewise-for-hsbm}, we already know that the entrywise bounds hold in these regimes. Therefore, it remains to show that each entry of $W u_2^*/\lambda_2^*$ is sufficiently bounded away from zero with high probability. In order to achieve this, we show the following lemma, which is similar in spirit to Lemma \ref{lem:chernoff bound for hsbm} but also captures denser regimes.
\begin{lemma}\label{lem:chernoff bound for hsbm-dense} 
Let $d \in \{2, 3, \dots\}$. Let $p_n$ and $q_n$ be parameterized according to \eqref{eq:parameter-regimes} for some $f_n$ and constants $\alpha > \beta > 0$. Let $W=\mathcal{S}(G)$ where $G\sim \hsbm(d,n,\alpha f_n,\beta f_n)$. If $f_n=\omega(\log n/n^{d-1})$, then there exists a constant $\eps:=\eps(d,\alpha,\beta)>0$ such that for any fixed $i\in [n]$, with probability at least $1-O(n^{-4})$,
$$ \sum_{j\in [n]} W_{ij} \sigma^*(i) \sigma^*(j) \geq \eps \cdot n^{d-1} f_n.$$
\end{lemma}
\begin{proof}
Let $X \triangleq  \sum_{j\in [n]} W_{ij} \sigma^*(i) \sigma^*(j)$. Define
\begin{align*}
N_r &= \binom{n/2}{r} \binom{n/2 -1}{d-1-r} \\
q_r &= \begin{cases}
\alpha f_n, & \text{if } r = 0\\
\beta f_n, & \text{if } 1 \leq r \leq d-1
\end{cases}\\
h_r &= d-1 -2r.
\end{align*}
Using identical steps as in the proof of Lemma \ref{lem:chernoff bound for hsbm}, we can show
\begin{align}
&\prob{X \leq \eps n^{d-1} f_n} \leq \exp \left(t \eps n^{d-1}f_n - \sum_{r=0}^{d-1} N_r q_r (1-e^{-h_rt})  \right). \label{eq:Chernoff-other-regimes}
\end{align}
Further analyzing the asymptotic behavior of $N_r$, we see that
\begin{align*}
N_r &= (1+o(1)) \frac{1}{2^{d-1}} \binom{d-1}{r} \binom{n-1}{d-1}=(1+o(1)) \frac{\binom{d-1}{r}  n^{d-1}}{2^{d-1} (d-1)!}.
\end{align*}
Substituting into \eqref{eq:Chernoff-other-regimes}, we obtain 
\begin{align}
&\mathbb{P}\left(X \leq \epsilon \, n^{d-1} f_n \right)\nonumber\\
&\leq \exp\left\{\left[ t \eps  - \frac{(1+o(1))}{2^{d-1} (d-1)!} \bigg(  \alpha (1-e^{-(d-1)t}) +\sum_{r=1}^{d-1}  \beta \binom{d-1}{r}(1-e^{-(d-1-2r)t})\bigg) \right] n^{d-1} f_n\right\} .\label{eq: final-chernoff}
\end{align}
Letting $t = t^{\star}(d, \alpha, \beta) > 0$ as in the proof of Lemma \ref{lem:chernoff bound for hsbm}, we obtain
\begin{align*}
   \prob{X \leq \eps \, n^{d-1} f_n} & \leq \text{exp} \left( \left(t^* \eps - \frac{I(d,\alpha,\beta)}{(d-1)!} + o(1) \right) n^{d-1}f_n \right). 
\end{align*}
Therefore, for $\eps = \eps(d,\alpha,\beta)$ sufficiently small, there exists $\delta = \delta(d,\alpha,\beta) > 0$ such that 
$$ \prob{X \leq \eps \, n^{d-1} f_n} \leq e^{-\delta \, n^{d-1} f_n}$$
for $n$ sufficiently large.
Finally, using $f_n=\omega(\log n/n^{d-1})$, we obtain
$$\prob{ \sum_{j\in [n]} W_{ij} \sigma^*(i) \sigma^*(j) \leq \eps \, n^{d-1} f_n } \leq e^ { -\delta \cdot \omega (\log n) } \leq e^{-4 \log n} =O(n^{-4}). \qed$$

\end{proof}
We now combine this with entrywise bounds on the eigenvector $u_2$ in Corollary \ref{cor:ewise-for-hsbm} to show our theorem.
\begin{proof}[Proof of Theorem \ref{theorem:spectral-otheregimes}]
    Since $f_n= \omega(\log n/n^{d-1})$, Corollary \ref{cor:ewise-for-hsbm} holds for some $s^*\in \{\pm 1\}$. Fixing $s=s^*$, and using the corollary we get that with probability $1-O(n^{-3})$,
\begin{equation}\label{eq:signs match dense}
   \sqrt{n} \min_{i \in [n]} s \sigma^*(i) u_{2,i} \geq   \sqrt{n} \min_{i \in [n]} s^2 \sigma^*(i) (W u_2^*)_i/\lambda_2^* - c(\log \log n)^{-1},  
\end{equation}
where $c$ is the constant from Corollary \ref{cor:ewise-for-hsbm}. As $s\in \{\pm 1 \}$, we have that $s^2=1$. Also, using \eqref{eq:sigma^* entries}, 
\begin{align*}
    \sqrt{n}\sigma^*(i)(Wu_2^*)_i = (1+o(1))  \sum_{j\in [n]} W_{ij} \sigma^*(i) \sigma^*(j). 
\end{align*}
By Lemma \ref{lem:chernoff bound for hsbm-dense}, since $\alpha>\beta>0$, there exists a positive constant $\eps(d,\alpha,\beta)>0$ such that for a fixed $i \in [n]$, $ \sum_{j\in [n]} W_{ij} \sigma^*(i) \sigma^*(j) \geq \eps \, n^{d-1} f_n$ with probability $1-O(n^{-4})$. Therefore, taking a union bound, we obtain that with probability $1-O(n^{-3})$, 
\[\sqrt{n} \min_{i \in [n]} s^2 \sigma^*(i) (W u_2^*)_i \geq (1+o(1))\epsilon \, n^{d-1} f_n.\]
Finally, note that $\lambda_2^* \asymp n^{d-1} f_n$. Therefore, \eqref{eq:signs match dense} implies that with probability $1-O(n^{-3})$
\begin{align*}
\sqrt{n} \min_{i \in [n]} s^* \sigma^*(i) (u_2)_i &\geq   (1+o(1)) \, \eps \, n^{d-1} f_n/\lambda_2^* - c(\log \log n)^{-1}> \eta,
\end{align*}
for some $\eta(d,\alpha,\beta)>0$, yielding the desired result.
\end{proof}

\subsection{Entrywise analysis}\label{subapp:ewise analysis}
We begin by recalling the setup of Theorem \ref{thm:entrywise-analysis}. For simplicity, let $\lambda = \lambda_k$, $\lambda^* = \lambda^*_k$ and $\Delta^*=\Delta_k^*$, dropping the subscript $k$. Let $s= \text{sgn}( \ip{u}{u^*})$, so that $\langle su, u^{*} \rangle \geq 0$. Also, for any fixed $m \in [n]$, let $G^{(m)}$ denote the hypergraph formed from $G$ by deleting all the edges incident on $m$. Let $\A{m}=\cS(G^{(m)})$, and let $(\lambda^{(m)}, \um{m})$ be th $k$-th eigenpair of $\A{m}$. Let $\sm{m}=\sign{\ip{\um{m}}{u^*}}$, so that $\ip{\sm{m} \um{m}}{u^*} \geq 0$. The notation $\lesssim$ and $\asymp$ hide constants in $d,c_0$, and $c_1$ (defined as in Theorem \ref{thm:entrywise-analysis}) throughout this section. Before proving the theorem, we require some additional observations and results. We begin by making a simple observation about the deterministic matrix $W^*$.
  \begin{observation}\label{obs:incoherence}
  $\twotoinfnorm{W^*} \leq \sqrt{n}\binom{n-2}{d-2}\mu_n.$ 
  \end{observation}
  \begin{proof}
  Recall that $\max_e p_e \leq \mu_n$. Since each entry $W_{ij}$ is a sum of $\binom{n-2}{d-2}$ Bernoulli random variables,  $\Exp{[W_{ij}]}=W^*_{ij} \leq \binom{n-2}{d-2}\mu_n$. Therefore,
  $$\twotoinfnorm{W^*}=\max_i \twonorm{W^*_{i\cdot}} \leq \sqrt{ n } \max {|W^*_{ij}|} \leq \sqrt{n}\binom{n-2}{d-2}\mu_n. \qed$$ 
  
  \end{proof}
Define the function $\varphi: \IR_{+} \rightarrow \IR_{+}$ such that
\begin{equation}
\varphi(x) = \frac{2+8d/c_0}{\left(1\vee \log(1/x)\right)}, \label{eq:varphi-definition}
\end{equation}
where $c_0$ is as in the statement of Theorem \ref{thm:entrywise-analysis}. We note the following properties of $\varphi(\cdot)$.
\begin{observation}\label{obs:monotonicity of phi}
$\varphi(x)$ is non-decreasing and $\varphi(x)/x$ is non-increasing on $\IR_{+}$.
\end{observation} 
The next result provides a probabilistic upper bound on the inner product of a row of $W-W^*$ and a fixed vector $v$, in terms of the function $\varphi(\cdot)$. Intuitively, one can think that the rate of growth of $\varphi(\cdot)$ essentially controls the strength of the concentration bound. Formally,
\begin{lemma}[\textbf{Row concentration}]\label{lem:row-conc}For any $m\in[n]$ and any fixed non-zero $v \in \IR^n$,
$$ \prob{|(W-W^*)_{m\cdot} v | \leq  \infnorm{v}\varphi \left( \frac{\twonorm{v}}{\sqrt{n} \infnorm{v}}\right) n \binom{n-2}{d-2} \mu_n} \geq 1-O\left(\frac{1}{n^4}\right).$$
\end{lemma}
 \begin{proof}Recall the definition of $\Em{m}=\{e \in \mathcal{E} : m \in e \}$. Let $v\in \mathbb{R}^n$ be a fixed vector. Let $X \triangleq (W-W^*)_{m \cdot} v $. It is convenient to rewrite $X$ as a sum of random variables $\{A_e\}_{e \in \mathcal{E}^{(m)}}$, where $A_e$ is the indicator random variable associated with the hyperedge $e$:
$$X=\sum_{e\in \Em{m}} \left(\sum_{j\in e\setminus{\{m\}}}v_j\right) (A_e-\Exp[A_e]) .$$

Without loss of generality, assume that $\infnorm{v}\hspace{-1mm}=\frac{1}{d}$ (otherwise, $v$ may be scaled). By Markov's inequality, for any $\delta,t>0$,
\begin{equation}
\prob{X \geq \delta } \leq  \prob{e^{t X} \geq e^{t \delta}} \leq e^{-t \delta} \prod_{e\in \Em{m}} \Exp\left(e^{t  (\sum_{j\in e\setminus{\{m\}}}v_j)(A_e-\Exp[A_e])} \right). \label{eq:row-concentration-Markov} 
\end{equation}
For $e \in \mathcal{E}^{(m)}$, we can bound the logarithm of the moment generating function as follows: 
\begin{align}
    &\log \left(\Exp\left[e^{t  (\sum_{j\in e\setminus{\{m\}}}v_j)(A_e-\Exp[A_e])} \right]\right) \nonumber\\
    &=\log \left(\Exp\left[e^{t  (\sum_{j\in e\setminus{\{m\}}}v_j)A_e} \right]\right) - t\left(\sum_{j\in e\setminus{\{m\}}}v_j \right) \mathbb{E}[A_e] \nonumber\\
    &=\log \left(1-p_{e}+p_{e}e^{t \sum_{j\in e\setminus{\{m\}}}v_j}\right) -t p_{e} \left(\sum_{j\in e\setminus{\{m\}}}v_j \right) \nonumber\\
    &\leq p_{e} \left(e^{t \sum_{j\in e\setminus{\{m\}}}v_j}-1\right) - t p_{e} \sum_{j\in e\setminus{\{m\}}}v_j,\label{eq:row-concentration-intermediate} 
    \end{align}
    where we have used the fact that $\log(1+x) \leq x$ for $x > 1$ in the last step. Next, we use the fact that $e^x \leq 1 + x + \frac{x^2}{2} e^r$ for $|x| \leq r$ to further upper-bound \eqref{eq:row-concentration-intermediate} by 
    \begin{align*}
    &p_{e}\left( 1+t \sum_{j\in e\setminus{\{m\}}}v_j+\frac{e^{t d\infnorm{v}}}{2} \cdot t^2 \left(\sum_{j\in e\setminus{\{m\}}}v_j\right)^2-1 \right) -t  p_{e} \sum_{j\in e\setminus{\{m\}}}v_j\\
    &= p_{e}\frac{e^{t}}{2} \cdot t^2 \left(\sum_{j\in e\setminus{\{m\}}}v_j\right)^2  
    \leq \hspace{1mm} \frac{e^t p_{\max} t^2}{2} \cdot d \sum_{j\in e\setminus{\{m\}}}{v_j}^2, 
\end{align*}
where $\max_{e} p_{e} \triangleq p_{\max}$.

Substituting our bounds on the log of moment generating functions into \eqref{eq:row-concentration-Markov}, we obtain
\begin{align*}
\log(\prob{X \geq \delta}) &\leq -t \delta + \frac{ e^{t} p_{\max} d }{2}t^2 \sum_{e \in \Em{m}} \sum_{j\in e\setminus{\{m\}}}{v_j}^2\\
&\leq -t \delta + \frac{e^{t}p_{\max}d}{2} \cdot t^2 \binom{n-2}{d-2} \norm{v}_2^{2},    
\end{align*}
where the last step follows from the fact that each $j\neq m$ appears in $\binom{n-2}{d-2}$ potential hyperedges in $\Em{m}$. Let $t=1 \bigvee \log \left( \frac{\sqrt{n}}{d \norm{v}_2}\right)$. Using the fact that $(1 \lor \log x)^2 \leq x$ for $x \geq 1$,
\begin{align*}
\log(\prob{X\geq \delta}) &\leq -t \delta +\frac{e^t p_{\max} d}{2}   \frac{\sqrt{n}}{d \norm{v}_2} \binom{n-2}{d-2} \norm{v}_2^{2}.
\end{align*}
Observe that $\twonorm{v} \leq \sqrt{n} \infnorm{v} = \sqrt{n}/d$, so that $\log\left(\frac{\sqrt{n}}{d \Vert v \Vert_2}\right) \geq 0$. Therefore, $e^{t}=e^{1 \vee \log(\frac{\sqrt{n}}{d \Vert v \Vert_2})} \leq  e^{1+\log(\frac{\sqrt{n}}{d \Vert v \Vert_2})} \leq \frac{e\sqrt{n}}{d \Vert v \Vert_2}.$ Hence
$$ \log(\prob{X\geq \delta}) \leq -t \delta +\frac{e p_{\max} d}{2} \cdot \frac{\sqrt{n}}{d \norm{v}_2}  \frac{\sqrt{n}}{d \norm{v}_2} \binom{n-2}{d-2} \norm{v}_2^{2} = -t \delta+\frac{e p_{\max}n \binom{n-2}{d-2}}{2d}.  $$
Let $a=8d/c_0$ and set $\delta=~t^{-1}d^{-1}(2+a)p_{\max} n \binom{n-2}{d-2}.$
We then obtain the bound
$$ \log(\prob{X\geq \delta}) \leq -\frac{(2+a) p_{\max} n}{d} \binom{n-2}{d-2}+\frac{e p_{\max} n \binom{n-2}{d-2}}{2d} \leq -\frac{a p_{\max} n \binom{n-2}{d-2}}{d}.$$
By replacing $v$ with $-v$, we obtain a similar bound for the lower tail. The union bound gives
$$\prob{|X| \geq \delta} \leq 2\exp\left(-\frac{a p_{\max} n \binom{n-2}{d-2}}{d}\right).$$
Substituting in the value of $t$ and using Assumption \ref{asump:well-separated} that $ n \binom{n-2}{d-2}\mu_n \geq c_0 \log n$,
$$\prob{|(W-W^*)_{m \cdot} v | \geq \frac{(2+a) n \binom{n-2}{d-2} \mu_n}{d\left(1 \bigvee \log \left( \frac{\sqrt{n}}{d \norm{v}_2}\right)\right)} } \leq 2e^{-\frac{ a {c_{0}} \log n }{d}}.$$
Finally, substituting the value of $a$,
$$\prob{|(W-W^*)_{m\cdot} v | \leq \frac{(2+8d/c_0) n \binom{n-2}{d-2} \mu_n}{d \left(1 \bigvee \log \left( \frac{\sqrt{n}}{d \norm{v}_2}\right) \right)} } \geq 1-2{n}^{-4}.$$
Recalling that $\infnorm{v} = 1/d$ yields:
$$\prob{|(W-W^*)_{m\cdot} v | \leq \frac{(2+\frac{8d}{c_0}) \infnorm{v} n \binom{n-2}{d-2} \mu_n}{1 \bigvee \log \left( \frac{\sqrt{n}\infnorm{v}}{\norm{v}_2}\right)} } \geq 1-2{n}^{-4}.$$
Since $\varphi\left( \frac{\twonorm{v}}{\sqrt{n}\infnorm{v}}\right)=(2+8d/c_0)/(1\vee \log \left( \frac{\sqrt{n}\infnorm{v}}{\norm{v}_2}\right)$, the lemma follows.\end{proof}

  \begin{lemma}\label{lem:combinatorial fact}
  $\twonorm{W-\A{m}} \lesssim \hspace{1mm}\twotoinfnorm{W}$.
  \end{lemma}
 \begin{proof}
  Note that $W-\A{m}$ is the similarity matrix of the graph with edges present only from $\Em{m}$. Consider any $i\neq m$. Since for any $e\in \Em{m}$ with $\{i,m\} \subseteq e$, the edge $e$ contributes to exactly $(d-1)$ entries of the $i$-th row of $W-\A{m}$, we have 
  \begin{equation}\label{eq:rows of A-Am}
      \Vert (W-\A{m})_{i \cdot} \Vert_1 = (d-1) W_{im}.
  \end{equation} Therefore $\twonorm{(W-\A{m})_{i\cdot}} \leq \norm{(W-\A{m})_{i\cdot}}_1 \leq (d-1)W_{im}$. 
  Hence we get 
  \begin{align*}
  \Vert W-\A{m} \Vert_2 &\leq \Vert W-\A{m} \Vert_F \\
  &= \sqrt{ \sum_{i \in [n]} \twonorm{(W-\A{m})_{i\cdot}}^2 } \\
  &\leq \sqrt{ \twonorm{W_{m\cdot}}^2+ \sum_{i\in [n]\setminus \{m\}} (d-1)^2 W_{im}^2} \\
  &\lesssim \hspace{1mm} \sqrt{\twonorm{W_{m\cdot}}^2} \hspace{1mm}\leq \hspace{1mm}\twotoinfnorm{W}.
  \end{align*}
  
\end{proof}
We now record a sharp spectral norm concentration result of Lee, Kim, and Chung \cite{Lee2020}, which will play a crucial role in our further analysis.
\begin{lemma}\cite[Special case of Theorem 4]{Lee2020}\label{lem:strong-spectralnorm}
    Fix $d\in\{2,3,\dots\}$. Let $p \in [0,1]^{\binom{[n]}{d}}$ be such that Assumption \ref{asump:well-separated} holds. Let $G\sim H(d,n,p)$ and $W=\mathcal{S}(G)$. Then there exists a constant $C=C(d,c_0)>0$ such that
    $$\prob{\twonorm{W-W^*} \leq C \sqrt{n \binom{n-2}{d-2} \mu_n} } \geq 1-O(n^{-11}).$$
\end{lemma}
Let us define a parameter $\gamma$, which controls the concentration in the analysis from here.
$$\gamma=\gamma_n \coloneqq \frac{C}{\sqrt{n\binom{n-2}{d-2} \mu_n}} \bigvee \frac{1}{\sqrt{n}},$$ where $C = C(d, c_0)>0$ is the constant from Lemma \ref{lem:strong-spectralnorm}. 
Recalling the definition of $\varphi(\cdot)$ (Equation \ref{eq:varphi-definition}), observe that
\begin{equation}\label{eq: gamma,phi(gamma)=o(1)}
    \gamma=o(1) \hspace{2 mm}\text{and } \hspace{2 mm} \varphi(\gamma) \lesssim \left(\frac{1}{1\vee \log\sqrt{n\binom{n-2}{d-2}\mu_n}} \bigvee \frac{1}{1\vee \log\sqrt{ n}} \right) \lesssim \frac{1}{\log \log n}=o(1),
\end{equation}
where we used Assumption \ref{asump:well-separated} that $n\binom{n-2}{d-2} \mu_n \geq c_0 \log n$.
We define the following event:
\[F_0:=\left\{ \twonorm{W-W^*} \leq \gamma \cdot n \binom{n-2}{d-2} \mu_n\right\}.\] 
By Lemma \ref{lem:strong-spectralnorm},
\begin{align*}
 \prob{F_0^{\textnormal{c}}}&=\prob{\twonorm{W-W^*}>\gamma \cdot n \binom{n-2}{d-2} \mu_n} \\
 &\leq \prob{ \twonorm{W-W^*} > C \sqrt{n \small{\binom{n-2}{d-2}} \mu_n}} \tag{by definition of $\gamma$}\\
 &=O(n^{-11}). \tag{using Lemma \ref{lem:strong-spectralnorm}}    
\end{align*}
Therefore, $\prob{F_0}\geq 1-O(n^{-11})$.
 We now derive some bounds on important quantities conditioned on the above event.
 \begin{lemma}\label{lem:lambda=lambda^*=log n}
Conditioned on $F_0$, we have $|\lambda^*| \asymp |\lambda| \asymp n^{d-1} \mu_n$.
\end{lemma}
\begin{proof}
Conditioned on the event $F_0$, Weyl's inequality implies 
$$|\lambda-\lambda^*| \leq \twonorm{W-W^*} \leq \gamma \cdot  n \binom{n-2}{d-2} \mu_n= o(n^{d-1} \mu_n),$$ where the last inequality follows since $n \binom{n-2}{d-2}=\Theta({n^{d-1}})$ and $\gamma=o(1)$ . Therefore, $\lambda \in \lambda^* \pm o(n^{d-1} \mu_n)$. Thus, Assumption \ref{asump:well-separated} (i.e. $|\lambda^*| \asymp n^{d-1} \mu_n$)  then ensures that 
$|\lambda^*| \asymp |\lambda| \asymp n^{d-1} \mu_n$.\end{proof}
  \begin{lemma}\label{lem:norm-inequalities}
  Conditioned on $F_0$, $$\twotoinfnorm{W} \lesssim \gamma \cdot n^{d-1} \mu_n \hspace{2mm}\textnormal{ and for all } m \in [n], \hspace{2mm}\twonorm{\A{m}-W^*} \lesssim \gamma \cdot n^{d-1} \mu_n.$$
  \end{lemma}
  \begin{proof}
  By the triangle inequality, 
  \begin{align*}
       \twotoinfnorm{W} & \leq \twonorm{W-W^*} + \twotoinfnorm{W^*} \\
       & \leq \hspace{1mm} \left(\frac{C}{\sqrt{n \binom{n-2}{d-2} \mu_n }} + \frac{1}{\sqrt{n}}\right) n \binom{n-2}{d-2} \mu_n \tag{using Lemma \ref{lem:strong-spectralnorm} and Observation \ref{obs:incoherence}}\\
      & \lesssim \gamma \cdot n^{d-1} \mu_n \tag{by definition of $\gamma$ and $n\binom{n-2}{d-2}=\Theta(n^{d-1})$}
  \end{align*}
  Similarly, using the triangle inequality and Lemma \ref{lem:combinatorial fact},
  $$   \Vert \A{m}-W^* \Vert_2 \hspace{1mm}\leq \hspace{1mm} \twonorm{W-W^*} + \Vert W-\A{m} \Vert_{2} \hspace{1mm} \lesssim \hspace{1mm} \twonorm{W-W^*} + \twotoinfnorm{W} \lesssim  \gamma \cdot n^{d-1} \mu_n. \qed
  $$
  
  \end{proof}
  
  Having derived the above bounds, we now bound the $\ell_2$ and $\ell_\infty$ norms of $(\sm{m}\um{m}-u^*)$ and $(su-\sm{m}\um{m})$ using two variants of the Davis and Kahan sin$(\theta)$ theorem. For completeness, we include the less well-known variant here, which is a special case of \cite[Theorem 3]{Deng2021}.
  \begin{proposition}[\textbf{Generalized Davis and Kahan sin$(\theta)$ theorem} \cite{Deng2021}]\label{prop:DK1970}
   Let $M \in \IR^{n \times n}$ be a symmetric matrix and let $X$ be the matrix that has the eigenvectors
of $M$ as columns. Then $M$ can be decomposed as $M={X} \Lambda {X^\top} = X_1 \Lambda_1 {X_1}^{\top} + X_2 \Lambda_2 {X_2}^{\top}$, where $X=[X_1 \hspace{1mm} X_2 ]$ and $\Lambda=\begin{bmatrix}
\Lambda_1 & 0\\
0 & \Lambda_2 
\end{bmatrix}.$
Suppose $\delta = \min_i | (\Lambda_2)_{ii} - \hat{\lambda}|$ is the absolute separation of some $\hat{\lambda}$ from $\Lambda_2$, then for any vector $\hat{u}$ we have
$$ \sin (\theta) \leq \frac{\Vert (M-\hat{\lambda} \mathbf{I})\, \hat{u} \Vert_2 }{\delta},$$
where $\theta$ is the canonical angle between the span of $X_1$ and $\hat{u}$.
\end{proposition}

 \begin{lemma}\label{lem:all diffs}
Conditioned on $F_0$,
     \begin{equation}\label{eq:um-u*_2}
        \max_{m\in [n]} \Vert \sm{m}\um{m}-u^* \Vert_2 \lesssim \gamma,
    \end{equation}
    \begin{equation}\label{eq:su-smum_2}
        \max_{m\in [n]} \Vert su-\sm{m}\um{m} \Vert_2 \lesssim (\gamma \wedge \infnorm{u}),
    \end{equation}
    \begin{equation}\label{eq:um-u*_inf}
        \max_{m\in [n]} \Vert \sm{m}\um{m}-u^* \Vert_{\infty} \lesssim (\infnorm{u} + \infnorm{u^*}).
    \end{equation}
 \end{lemma}
 \begin{proof}
 For any $m\in [n]$, we apply a variant of the Davis--Kahan sin$(\theta)$ theorem \cite[Corollary 3]{Yu2014AUV} to get
 $$ \Vert \sm{m}\um{m}-u^* \Vert_2 \hspace{1mm}\leq \hspace{1mm} \frac{2^{3/2} \twonorm{\A{m} -W^*}}{\Delta^*} \hspace{1mm}\lesssim \hspace{1mm}\frac{\gamma \cdot n^{d-1} \mu_n}{ n^{d-1} \mu_n} = \gamma,$$
 where the second inequality follows from Lemma \ref{lem:norm-inequalities} and Assumption \ref{asump:well-separated}, concluding the proof of \eqref{eq:um-u*_2}. Similarly, we also get
 \begin{equation}\label{eq:temp1}
     \Vert su-\sm{m}\um{m} \Vert_2 \hspace{1mm}\leq \hspace{1mm} \Vert su-u^* \Vert_2 + \Vert \sm{m}\um{m}-u^* \Vert_2  \hspace{1mm} \lesssim \hspace{1mm} \frac{\twonorm{W -W^*}}{\Delta^*} +\gamma \hspace{1mm}\lesssim \hspace{1mm} \gamma.
 \end{equation}
Therefore, $\ip{su}{\sm{m}\um{m}}\geq 0$. Let $\theta$ denote the angle between $su$ and $s^{(m)}u^{(m)}$. Then
\begin{align*}
    \Vert su-\sm{m}\um{m} \Vert_2 \leq \sqrt{ \Vert u \Vert_2^2 + \Vert \um{m} \Vert_2^2 - 2 \Vert u \Vert_2 \Vert \um{m} \Vert_2 \,\textnormal{cos }\theta} \leq \sqrt{2-2\,\textnormal{cos}^2 \theta} = \sqrt{2} \sin \theta.
\end{align*}
We then apply Proposition \ref{prop:DK1970} with $M=\A{m}, X_1=[\um{m}]$ and $(\hat{\lambda},\hat{u})=(\lambda,u)$: 
\begin{align}
    \Vert su-\sm{m}\um{m} \Vert_2 \leq \sqrt{2} \sin \theta &\leq \frac{\sqrt{2}\twonorm{(\A{m}-\lambda \mathbf{I})u}}{|\lambda^{(m)}_{k+1} - \lambda | \wedge | \lambda-\lambda^{(m)}_{k-1}|  } \nonumber \\
    &\leq \frac{\sqrt{2}\twonorm{(\A{m}-W)u}}{\left(\lambda_{k+1} - \lambda_k \right) \wedge \left(\lambda_k-\lambda_{k-1} \right)-\Vert W-\A{m} \Vert_2}  \nonumber\\
    & \lesssim \frac{\twonorm{(W-\A{m})u}}{\Delta^*-2\Vert W-W^* \Vert_2-\Vert W-\A{m} \Vert_2}.\nonumber
    \end{align}
Using Assumption \ref{asump:well-separated}, the definition of $F_0$, and Lemmas \ref{lem:combinatorial fact}, \ref{lem:norm-inequalities}, we can lower-bound the denominator by
\[\Delta^*-2\Vert W-W^* \Vert_2-\Vert W-\A{m} \Vert_2 \gtrsim (1-\gamma) n^{d-1} \mu_n.\]
Since $\gamma = o(1)$ by \eqref{eq: gamma,phi(gamma)=o(1)}, we obtain
\begin{align*}
\Vert su-\sm{m}\um{m} \Vert_2 \lesssim    \frac{\Vert (W-\A{m})u \Vert_2}{n^{d-1}\mu_n}. 
\end{align*}
Let $v = (W - W^{(m)}) u$.
We have already seen that the $m$-th row of $W-\A{m}$ is the same as $W$, so that $v_m = \lambda u_m$. Therefore, we bound the $m$-th entry of $v$ and the rest of its entries separately. Formally,
 \begin{align*}
    |v_m| & = |[Wu]_m| \leq |\lambda| |u_m| \leq |\lambda| \infnorm{u},\\
   |v_i| & =  |[(W-\A{m})u]_i| \leq \Vert (W-\A{m})_{i\cdot} \Vert_1 \infnorm{u} \lesssim W_{im} \infnorm{u}, \textnormal{ for $i\neq m$}, 
 \end{align*}
where the last inequality follows from \eqref{eq:rows of A-Am}. Therefore, 
\begin{equation}\label{eq:helpful2}
    \twonorm{v} \hspace{1mm}\lesssim \hspace{1mm}\infnorm{u} \sqrt{ \lambda^2 + \sum_{i\neq m} W_{im}^2} \hspace{1mm}\lesssim \infnorm{u}\hspace{1mm} \sqrt{ {|\lambda^*|}^2 + \twotoinfnorm{W}^2},
\end{equation}
where the last step follows from Lemma \ref{lem:lambda=lambda^*=log n}. Substituting the bound \eqref{eq:helpful2}:
\begin{equation}\label{eq:temp2}
    \Vert su-\sm{m}\um{m} \Vert_2 \hspace{1mm}\lesssim \hspace{1mm}\frac{\twonorm{v}}{n^{d-1}\mu_n} \hspace{1mm}\lesssim \hspace{1mm} \frac{\infnorm{u}\sqrt{ {|\lambda^*|}^2 + \twotoinfnorm{W}^2}}{n^{d-1} \mu_n} \hspace{1mm}\lesssim \hspace{1mm} \infnorm{u},
\end{equation}
where the last inequality follows from  Lemma \ref{lem:norm-inequalities} and Assumption \ref{asump:well-separated}. Observe that \eqref{eq:temp1} and \eqref{eq:temp2} together imply \eqref{eq:su-smum_2}.
 Finally, to prove \eqref{eq:um-u*_inf}, we apply the triangle inequality again and use \eqref{eq:temp2}.
 $$\Vert \sm{m}\um{m}-u^* \Vert_\infty \leq \Vert su- \sm{m}\um{m} \Vert_2 + \infnorm{u}+\infnorm{u^*} \lesssim \infnorm{u}+\infnorm{u^*}, $$ 
 concluding the proof. \end{proof}
 \begin{lemma} \label{lem:probabilistic Au* bound}
 With probability at least $1-O(n^{-3})$,
$$     \infnorm{(W-W^*)u^*} \lesssim \infnorm{u^*} n^{d-1} \mu_n  \hspace{3mm}\textnormal{ and } \hspace{3mm} \infnorm{Wu^*} \lesssim  \infnorm{u^*} n^{d-1} \mu_n.$$
 \end{lemma}
 \begin{proof}
 To prove the first inequality, we apply Lemma \ref{lem:row-conc}. For each $m \in [n]$, we have 
 \[|(W-W^*)_{m\cdot} u^*|\leq \varphi\left(\frac{\twonorm{u^*}}{\sqrt{n} \infnorm{u^*}}\right)\infnorm{u^*} n \binom{n-2}{d-2} \mu_n\] with probability $1-O(n^{-4})$. Using the monotonicity of $\varphi$ (Observation \ref{obs:monotonicity of phi}) and $\twonorm{u^*} \leq \sqrt{n} \infnorm{u}$, we obtain 
 \[|(W-W^*)_{m\cdot} u^*|\lesssim \varphi(1) \infnorm{u^*} n^{d-1} \mu_n\] 
 with probability $1-O(n^{-4})$. Note that $\varphi(1)=O(1)$. Taking a union bound over all $m \in [n]$, we have  $\infnorm{(W-W^*)u^*} \lesssim \infnorm{u^*} n^{d-1} \mu_n  $, with probability at least $1-O(n^{-3})$.
 
 To get the second statement, we apply the first statement to show that with probability $1-O(n^{-3})$,
 $$\infnorm{Wu^*} \leq \infnorm{W^*u^*} +\infnorm{(W-W^*)u^*} \lesssim |\lambda^*| \infnorm{u^*} + \infnorm{u^*} n^{d-1} \mu_n \lesssim \infnorm{u^*} n^{d-1} \mu_n.$$
 The last inequality in the above follows from Assumption \ref{asump:well-separated}.\end{proof}
\begin{lemma}\label{lem:probabilistic Avm bound}
With probability at least $1-O(n^{-3})$,
 $$ \max_{m\in [n]} |W_{m\cdot} ( s^{(m)} \um{m} -u^*)| \lesssim (\gamma +\varphi(\gamma))(\infnorm{u}+\infnorm{u^*}) n^{d-1}\mu_n.$$
 \end{lemma}
 \begin{proof}
We denote $(\sm{m} \um{m} -u^*)$ by $\vm{m}$ for notational convenience. Recall that $\prob{F_0} \geq 1-O(n^{-11})$ by Lemma \ref{lem:chernoff bound for hsbm}. Conditioned on $F_0$, for all $m\in [n]$
 \begin{align}
     |W_{m\cdot} \vm{m}| &\leq |W^*_{m\cdot} \vm{m}| + |(W-W^*)_{m\cdot} \vm{m}| \tag{by the triangle inequality} \\
     & \leq \twotoinfnorm{W^*} \Vert \vm{m} \Vert_2 + |(W-W^*)_{m\cdot} \vm{m}| \tag{by the Cauchy--Schwarz inequality}\\
     & \lesssim \sqrt{n} \cdot n^{d-2}\mu_n \cdot \gamma + |(W-W^*)_{m\cdot} \vm{m}| \tag{by Observation \ref{obs:incoherence} and \eqref{eq:um-u*_2}}\\
     &=\frac{n^{d-1} \mu_n \, \gamma}{\sqrt{n}} + |(W-W^*)_{m\cdot} \vm{m}| \nonumber\\
     &\leq \gamma \infnorm{u^*} n^{d-1} \mu_n + |(W-W^*)_{m\cdot} \vm{m}|. \label{eq:main inequality}
 \end{align}
 We now focus on bounding the second term. We know that $\vm{m}$ is independent of the randomness in the $m$-th row of $W$. Therefore, by the row concentration result (Lemma \ref{lem:row-conc}), for a fixed $m\in [n]$
 \begin{equation}\label{eq:temp3}
    |(W-W^*)_{m\cdot} \vm{m}| \leq  \Vert \vm{m} \Vert_\infty \varphi\left( \frac{ \Vert \vm{m} \Vert_2}{\sqrt{n} \Vert \vm{m} \Vert_\infty}  \right) n \binom{n-2}{d-2} \mu_n , 
 \end{equation}
 holds with probability at least $1-O(n^{-4})$. Let $F$ be the event that \eqref{eq:temp3} holds simultaneously for all $m\in [n]$. By a union bound, $\prob{F} \geq 1-O(n^{-3})$. Conditioned on the event $F$, let us consider two different cases. \\
 
  \underline{\emph{Case 1}}: Suppose $\frac{ \Vert \vm{m} \Vert_2}{\sqrt{n} \Vert \vm{m} \Vert_\infty} \leq \gamma$. Under this case, we use the fact that $\varphi(x)$ is non-decreasing (Observation \ref{obs:monotonicity of phi}) in \eqref{eq:temp3} to get:
 $$ |(W-W^*)_{m\cdot} \vm{m}| \leq \varphi(\gamma) \Vert \vm{m} \Vert_{\infty} n \binom{n-2}{d-2} \mu_n.$$
\underline{\emph{Case 2}}: Suppose $\frac{ \Vert \vm{m} \Vert_2}{\sqrt{n} \Vert \vm{m} \Vert_\infty} > \gamma$. In this case, multiplying and dividing \eqref{eq:temp3} by $\frac{\Vert \vm{m} \Vert_2}{\sqrt{n}}$,
\begin{align*}
    |(W-W^*)_{m\cdot} \vm{m}| 
    &\leq \frac{\Vert \vm{m} \Vert_2}{\sqrt{n}} \frac{\sqrt{n}\Vert \vm{m} \Vert_\infty}{\Vert \vm{m} \Vert_2} \varphi\left( \frac{ \Vert \vm{m} \Vert_2}{\sqrt{n} \Vert \vm{m} \Vert_\infty}  \right) \cdot  n \binom{n-2}{d-2} \mu_n  \\
    & \leq  \frac{\varphi(\gamma)}{\gamma}\frac{\Vert\vm{m} \Vert_2}{\sqrt{n}} \cdot n \binom{n-2}{d-2} \mu_n, 
\end{align*}
where we have used the fact that $\varphi(x)/x$ is non-increasing (Observation \ref{obs:monotonicity of phi}).
Combining both cases, for all $m\in [n]$:
$$     |(W-W^*)_{m\cdot} \vm{m}| \leq \varphi(\gamma) \cdot n  \binom{n-2}{d-2} \mu_n \left( \Vert \vm{m} \Vert_\infty \vee \frac{\Vert \vm{m}\Vert_2}{\gamma \sqrt{n}} \right).$$
Substituting bounds from \eqref{eq:um-u*_2} and \eqref{eq:um-u*_inf} in the above, and then using the fact that $\infnorm{u^*} \geq \twonorm{u^*}/\sqrt{n}=1/\sqrt{n}$, we obtain
$$     |(W-W^*)_{m\cdot} \vm{m}| \lesssim \varphi(\gamma) \cdot n^{d-1} \mu_n \left( \left(\infnorm{u}+\infnorm{u^*} \right) \vee \frac{\gamma}{\gamma \sqrt{n}} \right) \lesssim \varphi(\gamma) \left( \infnorm{u}+ \infnorm{u^*}\right) n^{d-1} \mu_n.$$

Finally, we substitute into \eqref{eq:main inequality}:
\begin{align*}
\max_{m\in[n]} |W_{m\cdot} \vm{m}| &\lesssim \gamma \infnorm{u^*} n^{d-1} \mu_n + \varphi(\gamma)  (\infnorm{u}+\infnorm{u^*}) n^{d-1} \mu_n \\
&\leq (\gamma +\varphi(\gamma))(\infnorm{u}+\infnorm{u^*}) n^{d-1} \mu_n,    
\end{align*}
concluding the proof.
\end{proof}
We finally prove Theorem \ref{thm:entrywise-analysis}.
\begin{proof}[Proof of Theorem \ref{thm:entrywise-analysis}]
Recall the definition of the event $F_0$ and $\prob{F_0} \geq 1-O(n^{-11})$ (Lemma \ref{lem:strong-spectralnorm}). Conditioned on $F_0$, we have $|\lambda| \asymp |\lambda^*|$ by Lemma \ref{lem:lambda=lambda^*=log n}; we use this throughout the proof. 

We first bound $\infnorm{u}$ in terms of $\infnorm{u^*}$ as we need our final bounds only in terms of the latter. By the triangle inequality,
\begin{align*}
        \infnorm{u} &= \infnorm{\frac{sWu}{\lambda}} \hspace{-1mm} \leq  \hspace{1mm}\infnorm{\frac{Wu^*}{\lambda}} \hspace{-1.5mm} + \infnorm{\frac{W(su-u^*)}{\lambda}} \lesssim \hspace{1mm} \frac{1}{|\lambda^*|} \left(\infnorm{Wu^*} + \max_{m \in [n]} {|W_{m\cdot}(su-u^*)|} \right)\\
        & \leq \hspace{1mm} \frac{1}{|\lambda^*|} \left(\infnorm{Wu^*} + \max_{m \in [n]} |W_{m\cdot}(su-\sm{m}\um{m}) | + \max_{m\in [n]} |W_{m\cdot}(\sm{m}\um{m}-u^*) | \right) \\
        & \leq \hspace{1mm} \frac{1}{|\lambda^*|} \left(\infnorm{Wu^*} + \twotoinfnorm{W} \max_{m\in [n]}\Vert su-\sm{m}\um{m} \Vert_{2}+ \max_{m \in [n]} | W_{m\cdot}(\sm{m}\um{m}-u^*) | \right).
\end{align*}
We substitute the derived bounds from Lemma \ref{lem:probabilistic Au* bound} in the first term, Lemmas \ref{lem:norm-inequalities} and \ref{lem:all diffs} in the second term,  and Lemma \ref{lem:probabilistic Avm bound} in the third term. In particular, we obtain that with probability at least $1-O(n^{-3})$,
\begin{align*}
   \infnorm{u}  & \lesssim \frac{1}{|\lambda^*|} \Big(\infnorm{u^*} n^{d-1} \mu_n + (\gamma \cdot n^{d-1} \mu_n) \infnorm{u} + (\gamma + \varphi(\gamma))(\infnorm{u}+\infnorm{u^*}) n^{d-1} \mu_n \Big)\\
    & \lesssim \Big((1+\gamma+\varphi(\gamma)) \infnorm{u^*}+ (\gamma+\varphi(\gamma) \infnorm{u} \Big). \hspace{10mm} \tag{using Assumption \ref{asump:well-separated}}
    \end{align*}
    Using the fact that $\gamma,\varphi(\gamma)=o(1)$ from \eqref{eq: gamma,phi(gamma)=o(1)}, we have that for some constant $c_3>0$,
  \begin{align}
 \infnorm{u}   & \leq c_3 \infnorm{u^*}+o(1) \infnorm{u} \nonumber\\
  (1-o(1)) \infnorm{u}  & \leq  c_3 \infnorm{u^*} \nonumber\\
   \infnorm{u} & \leq \hspace{2mm} \frac{c_3}{1-o(1)}\infnorm{u^*} \lesssim \infnorm{u^*} \label{eq:u_inf bound}. 
  \end{align}  
  
Similarly, conditioned on $F_0$, we bound the quantity of interest. Using the triangle inequality,
\begin{align*}
         \infnorm{su-\frac{Wu^*}{\lambda^*}} &= \infnorm{\frac{sWu}{\lambda} - \frac{Wu^{*}}{\lambda} + \frac{Wu^{*}}{\lambda} -\frac{Wu^*}{\lambda^*}}\\
         & \leq \left| \frac{1}{\lambda} -\frac{1}{\lambda^*} \right| \infnorm{Wu^*}+ \frac{1}{|\lambda|} \infnorm{W(su-u^*)}.
         \end{align*}
Recalling that $|\lambda| \asymp |\lambda^*|$,
         \begin{align*}
\infnorm{su-\frac{Wu^*}{\lambda^*}}  &\lesssim  \frac{1}{|\lambda^*|}\left(\gamma\infnorm{Wu^*} + \infnorm{W(su-u^*)}\right)\\
       \lesssim  \frac{1}{|\lambda^*|}  &\left(\gamma\infnorm{Wu^*} + \twotoinfnorm{W} \max_{m\in [n]}\Vert su -\sm{m}\um{m} \Vert_2 + \max_{m\in [n]} | W_{m\cdot} (\sm{m}\um{m}-u^*)| \right).
\end{align*}
 We again substitute the bounds from Lemma \ref{lem:probabilistic Au* bound} in the first term, Lemma \ref{lem:norm-inequalities} and \ref{lem:all diffs} in the second term,  and Lemma \ref{lem:probabilistic Avm bound} in the third term. In particular, we obtain that with probability at least $1-O(n^{-3})$,
\begin{align*}
  \infnorm{su-\frac{Wu^*}{\lambda^*}}  & \lesssim \frac{1}{|\lambda^*|} \Big( \gamma \infnorm{u^*} n^{d-1}\mu_n + (\gamma \cdot n^{d-1} \mu_n) \infnorm{u} + (\gamma + \varphi(\gamma))(\infnorm{u}+\infnorm{u^*}) n^{d-1}\mu_n \Big)\\
    & \lesssim \Big((\gamma+\varphi(\gamma)) \infnorm{u^*}+ (\gamma+\varphi(\gamma) \infnorm{u} \Big) \tag{using Assumption \ref{asump:well-separated}}\\
    & \lesssim (\gamma+\varphi(\gamma)) \infnorm{u^*} \tag{using $\infnorm{u} \lesssim \infnorm{u^*}$ as shown in \eqref{eq:u_inf bound}}\\
   & \lesssim \frac{\infnorm{u^*}}{ \log \log n}. \tag{$\gamma=O(\varphi(\gamma))$ and using \eqref{eq: gamma,phi(gamma)=o(1)}} 
\end{align*}
This is precisely the bound we need, concluding the proof.
\end{proof}

\section{Spectral norm concentration for similarity matrices}\label{appendix:spectral-norm}
In order to prove Theorem \ref{thm:sp-norm-conc}, we require the following result.
\begin{lemma}\label{lem:sp-norm-i.i.d matrix}
Let $A$ be a random matrix with independent entries, where $A_{ij} \in [a,b]$ for two constants $a<b$. Suppose $\Exp{[|A_{ij}|]} \leq q$ for all $i,j$, where $\frac{c_2 \log n}{n} \leq q \leq 1-c_3$ for arbitrary constants $c_2,c_3>0$. Then, there exists a constant $c'\coloneqq c'(c_2,c_3,a,b)>0$ such that
$$\Exp{[\twonorm{A - \Exp{[A}]}]} \leq c' \sqrt{nq}.$$
\end{lemma}
\begin{proof} We use ideas from \cite[Lemma 4.5]{Dhara2022c}, who showed a similar result for zero-diagonal, symmetric matrices with independent entries. We first construct a symmetric matrix $B$ using $A$ to reduce it to the symmetric case.
Let $$B=\begin{bmatrix}
0 & A \\
A^\top & 0
\end{bmatrix}.$$
Fix a vector $x \in \IR^{n}$ such that $\twonorm{x}=1$, and consider the vector $y \in \IR^{2n}$ such that $y= \begin{bmatrix}
\mathbf{0}_n \\
x
\end{bmatrix}$. Observe that $\twonorm{y}=1.$ Moreover,
$$ (B-\Exp[B])y = \begin{bmatrix} 
(A-\Exp[A])x\\
\mathbf{0}_n
\end{bmatrix}.$$

Therefore,
$$\twonorm{(B-\Exp[B])y} = \twonorm{(A-\Exp[A]) x},$$
so that \[\Vert(A - \mathbb{E}[A])x \Vert_2 = \Vert(B - \mathbb{E}[B])y \Vert_2 \leq \Vert B - \mathbb{E}[B] \Vert_2.\] Since $x$ was arbitrary, we have shown $ \twonorm{A-\Exp{[A]}} \leq  \twonorm{B-\Exp{[B]}}$, and therefore
\begin{equation}
 \Exp[ \twonorm{A-\Exp{[A]}} ] \leq \Exp[\twonorm{B-\Exp{[B]}}].  \label{eq:AB}
\end{equation}
Thus, it suffices to bound $\Exp[\twonorm{B-\Exp{[B]}}]$. Let $B^+ = \max\{B, 0\}$, where the maximum is taken entrywise. Similarly, let $B^- = -\min\{B, 0\}$. Then we can write $B = B^+ - B^-$. By the triangle inequality,
\begin{equation}
\mathbb{E}\left[\Vert B - \mathbb{E}[B] \Vert_2 \right] \leq \mathbb{E}\left[\Vert B^+ - \mathbb{E}[B^+] \Vert_2 \right]  + \mathbb{E}\left[\Vert B^- - \mathbb{E}[B^-] \Vert_2 \right].  \label{eq:B+-}
\end{equation}
Observe that $B^+$ and $B^-$ are nonnegative, zero-diagonal, symmetric matrices with independent entries. Also, for all $i,j$,
$$ \max\{ \mathbb{E}[B^+_{ij}], \mathbb{E}[B^-_{ij}] \} \leq \Exp[|B_{ij}|] \leq \max \Exp[|A_{ij}|] \leq q.$$ 
If $b\leq 0$, then $ \twonorm{B^{+}-\Exp[B^+]}=0$. Otherwise it follows from \cite[Theorem 5]{Hajek2016} that there exists $c^{+}>0$ such that 
$$ \mathbb{E}\left[ \frac{1}{b} (\Vert B^+ - \mathbb{E}[B^+]\Vert_2) \right] \leq c^{+} \sqrt{\frac{nq}{b}}.$$
Similarly, if $a \geq 0$, then $ \twonorm{B^{-}-\Exp[B^-]}=0$. Otherwise, we again use \cite[Theorem 5]{Hajek2016} to conclude that there exists a constant $c^{-}>0$ such that
\[\mathbb{E}\left[ \frac{1}{|a|}(\Vert B^- - \mathbb{E}[B^-]\Vert_2) \right] \leq c^{-}\sqrt{\frac{nq}{|a|}}.\]
Combining these with \eqref{eq:AB} and \eqref{eq:B+-} we get
$$\Exp[\twonorm{A-\Exp{[A]}}] \leq c'\sqrt{nq} ,$$
where $c'= c^{+} \sqrt{\max\{b,0\}} + c^{-} \sqrt{|\min\{a,0\}|}$.
\end{proof}

\begin{proof}[Proof of Theorem \ref{thm:sp-norm-conc}]
The symbols $\lesssim$ and $ \asymp$ hide constants in $d,c_0$ throughout the proof. Our goal is to bound $\Exp{\left[\twonorm{\mathcal{S}(G)-\Exp{[\mathcal{S}(G)]}}\right]}$. Let $G'$ be an independent copy of $G$. Observe that for a fixed matrix $X$, the function $f(Y) = \Vert X - Y\Vert_2$ is convex. By Jensen's inequality,
\begin{align*}
\mathbb{E}\left[ \Vert \cS(G) - \mathbb{E}[\cS(G)] \Vert_2 \right] &= \mathbb{E}\left[ \Vert \mathcal{S}(G) - \mathbb{E}[\mathcal{S}(G')] \Vert_2 \right] \leq  \mathbb{E}\left[ \Vert \mathcal{S}(G) - \mathcal{S}(G') \Vert_2 \right].
\end{align*}
We can extend the definition of $\mathcal{S}$ so that $\mathcal{S}(G - G') = \mathcal{S}(G) - \mathcal{S}(G')$; i.e. $G - G'$ is a ``hypergraph'' with edges labeled by $\{1, 0,-1\}$.

Let $R$ be a symmetric tensor of order $d$ and dimension $n$ with independent Rademacher entries; i.e. the entries $\{R(i_1, i_2, \dots, i_d) : i_1 \leq i_2 \leq \cdots \leq i_d\}$ are mutually independent. Let $\circ$ denote the edge-wise product. Since $G - G'$ has the same distribution as $(G - G') \circ R$, we obtain
\begin{equation}\label{eq: sp-norm temp1}
\mathbb{E}\left[ \Vert \cS(G) - \mathbb{E}[\cS(G)] \Vert_2 \right] \leq  \mathbb{E}\left[ \Vert \mathcal{S}(G - G')  \Vert_2 \right] = \mathbb{E}\left[ \Vert \mathcal{S}((G - G') \circ R)  \Vert_2 \right]
\leq 2 \mathbb{E}\left[\Vert \mathcal{S}(G \circ R)  \Vert_2 \right], 
\end{equation}
where the last inequality follows from the triangle inequality. Let $p_{\max} \triangleq c_0 \log n/\binom{n-1}{d-1}$ for simplicity. Consider the hypergraph $G_{+}$ that is coupled to $G$ as follows. The hypergraph $G_{+}$ does not contain any edge that appears in $G$. Each edge $e$ that is not present in $G$ is present in $G_{+}$ with probability $\frac{p_{\max}-p_e}{1-p_e}$ (independently across edges). Letting $G^{(1)} \sim \hsbm(d,n,p_{\max},p_{\max})$, we see that $(G+G_{+})\circ R$ has the same distribution as $G^{(1)} \circ R$. Also, $\mathbb{E}\left[\mathcal{S}(G_{+} \circ R) ~|~ G \right] = 0$. Using these observations along with Jensen's inequality, we obtain
\begin{align*}
\mathbb{E}\left[\Vert \mathcal{S}(G \circ R)  \Vert_2 \right] &= \mathbb{E}\left[\Vert \mathcal{S}(G \circ R) + \mathbb{E}\left[\mathcal{S}(G_{+} \circ R)~|~G \right]  \Vert_2 \right]
\leq \mathbb{E}\left[\Vert \mathcal{S}(G \circ R) + \mathcal{S}(G_{+} \circ R)  \Vert_2 \right]\\
&= \mathbb{E}\left[\Vert \mathcal{S}((G+G_{+}) \circ R)  \Vert_2 \right]
= \mathbb{E}\left[\Vert \mathcal{S}(G^{(1)} \circ R)  \Vert_2 \right].
\end{align*}
Substituting this into \eqref{eq: sp-norm temp1}, we have
\begin{equation} \label{eq:sp norm temp 2}
\mathbb{E}\left[ \Vert \cS(G) - \Exp{[\cS(G)]} \Vert_2 \right] \leq 2 \mathbb{E}\left[\Vert S(G^{(1)} \circ R) \Vert_2 \right].    
\end{equation}
Note that, even though $\cS(G^{(1)} \circ R)$ has identically distributed entries, they are still dependent; we want a matrix with independent entries instead to apply Lemma \ref{lem:sp-norm-i.i.d matrix}. To achieve this, we use a somewhat involved symmetrization argument. For simplicity, consider the case when $d=2$ (which reduces to the adjacency matrix), where we have independent entries up to the symmetry. In other words, each independent edge random variable appears exactly twice in the matrix. In this situation, \cite[Theorem 5]{Hajek2016} uses the standard symmetrization technique by adding another independent copy of $\cS(G^{(1)} \circ R)$ and rearranging random variables to create independence. However, for a general $d$, this is not sufficient. But observe that the random variable (label) associated with any fixed edge $e$ is added to exactly $K:=2 \cdot \binom{d}{2} = d^2-d$ entries of the similarity matrix when we apply the map $\cS(\cdot)$. Thus, we instead add $K$ independent copies to create enough independence and show how to rearrange the hyperedge random variables such that each matrix has fully independent entries after the rearrangement. 

More formally, let $G^{(m)}$ and $R^{(m)}$ be independent copies of $G^{(1)}$ and $R$ respectively for $m \in [K]$. Note that $\Exp{[\cS(G^{(m)}\circ R^{(m)})]}$ is the zero matrix. Thus, adding the zero matrix, and then using Jensen's inequality we get
\begin{align}
  \Exp{[\Vert \cS(G^{(1)}\circ R) \Vert_2]}&=\Exp{\left[\twonorm{\cS(G^{(1)}\circ R^{(1)})+ \sum_{m=2}^{K} \Exp{[\cS(G^{(m)}\circ R^{(m)})]}}\right]} \nonumber \\
  &\leq \Exp \left[\twonorm{\sum_{m=1}^{K} \cS(G^{(m)}\circ R^{(m)})}\right]. \label{eq:sp norm temp 3} 
\end{align}

Observe that $\sum_{m=1}^{K} \cS(G^{(m)}\circ R^{(m)})$ is the sum of \emph{independent} copies of random matrices with \emph{dependent} entries. The goal is to re-express this same quantity as the sum of \emph{dependent} matrices with fully \emph{independent} entries. To this end, let us consider the following construction. Let $L(e)$ be a fixed-ordered list of locations to which the random variable associated with the edge $e$ is added. Formally, $L(e):=(\mathsf{op}^{(1)}_e,\dots, \mathsf{op}^{(K)}_e )$ is an ordered list of all ordered pairs of $e$; i.e. $(i,j): i,j \in e$. 

For $e \in \mathcal{E}$ and $m, \ell \in [K]$ let $X^{(e,m,\ell)}$ be the $n \times n$ matrix which has only one non-zero entry:
\begin{align*}
X_{ij}^{(e,m,\ell)} &= \begin{cases}
A^{(m)}_e\circ R^{(m)}_e  & (i,j)=\mathsf{op}^{(\ell)}_e\\
0 & \text{otherwise},
\end{cases}   
\end{align*}
where $A_e^{(m)}$ denotes the indicator random variable associated with the edge $e$ in $G^{(m)}$. By construction,
\begin{equation}\label{eq:sum of independent copies}
  \sum_{m=1}^{K} \cS(G^{(m)}\circ R^{(m)}) =\sum_{m=1}^{K} \sum_{e\in \mathcal{E}} \sum_{\ell=1}^{K} X^{(e,m,\ell)}   
\end{equation}
We will now re-express this summation as the sum of \emph{dependent} matrices with \emph{independent} entries. Let $Y=\sum_{i=1}^{\binom{n-2}{d-2}} X_i Z_i$, where $X_i\sim \textnormal{Bern}(p_{\max})$ and $Z_i \sim \text{Rad}$ are independent. Let $D$ be a diagonal matrix whose diagonal entries are i.i.d. copies of $Y$.
Let us consider matrices $C^{(1)},\dots,C^{(K)}$, where
$$C^{(1)}=\sum_{e\in \mathcal{E}} \sum_{\ell=1}^{K} X^{(e,\ell,\ell)}+D$$
$$C^{(k)}= \sum_{e\in \mathcal{E}} \sum_{\ell=1}^{K}  X^{(e,(\ell+k-1)\% K, \ell)}+ (-1)^{k+1} D $$
Here $\%$ represents the standard modulo operation, except $K\%K$ is $K$ instead of $0$. Intuitively, $C^{(1)}$ is the matrix for which the $K$ locations corresponding to a given edge ``consult'' $K$ independent copies of $G$. Observe that all the entries of $C^{(1)}$ are independent copies of $Y$. To ensure that we add the random variable associated with an edge of any given copy of $G$ at all $K$ locations in the similarity matrix, we iterate over the list in a cyclic manner when constructing $C^{(2)}, \dots, C^{(K)}$. Note that $D$ has the same distribution as $-D$. Therefore, from the symmetry all $C^{(k)}$s have the same distribution. Moreover,
\begin{align}
  \sum_{k=1}^{K} C^{(k)} &=  \sum_{k=1}^{K} \sum_{e\in \mathcal{E}} \sum_{\ell=1}^{K}  X^{(e,(\ell+k-1)\% K,\ell)}+ (-1)^{k+1} D  \nonumber\\
  &=  \sum_{e\in \mathcal{E}} \sum_{\ell=1}^{K} \sum_{k=1}^{K}  X^{(e,(\ell+k-1)\% K,\ell)}+ (-1)^{k+1} D  \nonumber\\
  &= \sum_{e\in \mathcal{E}} \sum_{\ell=1}^{K} \sum_{m=1}^{K}  X^{(e,m,\ell)}  \nonumber\\
  &= \sum_{m=1}^{K} \cS(G^{(m)}\circ R^{(m)}), \label{eq:sum of dependent copies}
\end{align}
where the last step follows from \eqref{eq:sum of independent copies}. Therefore, combining \eqref{eq:sp norm temp 2}, \eqref{eq:sp norm temp 3}, and \eqref{eq:sum of dependent copies}
\begin{equation}\label{eq:sp norm final}
   \Exp{[\twonorm{W-W^*}]} \leq   2\Exp{\left[\twonorm{\sum_{k=1}^{K}C^{(k)}} \right]} \leq 2K \Exp{[\twonorm{C}]},
\end{equation}
where $C$ has the same distribution as $C^{(1)}$. Each entry of $C$ is an independent copy of $Y$. Let $F$ be the event that all the entries of $C$ are in the range $[-4d,4d]$. By a union bound, 
\begin{align*}
    \prob{F^{\textnormal{c}}} &\leq \, n^2 \cdot \prob{|Y|\geq 4d } \leq \, n^2 \cdot \prob{ \sum_{i=1}^{\binom{n-2}{d-2}}{X_i} \geq 4d } .
 \end{align*}
Since $\sum_{i=1}^{\binom{n-2}{d-2}}{X_i} \sim \textnormal{Bin}\left( \binom{n-2}{d-2}, p_{\max}\right)$, by \cite[Theorem 4.4, Equation 4.1]{mitzenmacher2017probability}
 \begin{equation} \prob{F^\textnormal{c}} \leq n^2 \cdot \prob{\text{Bin}\left(\binom{n-2}{d-2}, \frac{c_0 \log n}{\binom{n-1}{d-1}} \right) \geq 4d } \leq n^2 \frac{e^{4d}}{\Theta(n/\log n)^{4d}} = O\left( \frac{1}{n^{4d-3}} \right). \label{eq:F-c}\end{equation}
Therefore,
\begin{align}
    \Exp{[\twonorm{C}]}&= \prob{F} \Exp{[\twonorm{C} \mid F]} + \prob{F^{\textnormal{c}}}\Exp{[\twonorm{C} \mid F^\textnormal{c}]} \nonumber\\
    &\leq \Exp{[\twonorm{C} \mid F]}+ O\left( \frac{1}{n^{4d-3}}\right) \Exp[\Vert C \Vert_F \mid F^\textnormal{c}] \nonumber \\
    & \leq \Exp{[\twonorm{C} \mid F]} + O\left( \frac{1}{n^{4d-3}}\right) n \binom{n-2}{d-2}
    \nonumber\\
    &= \Exp{[\twonorm{C} \mid F]} + O \left(n^{3 - 4d + 1 + d-2} \right) \nonumber\\
    & \leq \Exp{[\twonorm{C} \mid F]} +o(1),  \label{eq:conditional-expectation remaining part}
\end{align}
where the second inequality uses the fact that each entry of $C$ is at most $\binom{n-2}{d-2}$.
Thus, it is only left to show $\Exp{[\twonorm{C} \mid F]} = O(\sqrt{\log n}) $. Note that the entries of $C$ are independent even after conditioning on $F$. Moreover, the entries are bounded in $[-4d,4d]$. Thus
\begin{align*}
  \Exp{[|C_{ij}| \mid F]}&=\sum_{k=1}^{4d} k \cdot \prob{|C_{ij}|=k \mid F} 
  \leq \sum_{k=1}^{4d} k \cdot \frac{\prob{F \cap |C_{ij}|=k}}{\prob{F}}= \sum_{k=1}^{4d} k\cdot \frac{\prob{|C_{ij}|=k }}{1-o(1)} \\
  &\lesssim \sum_{k=1}^{4d} k \cdot  \prob{|C_{ij}|=k} \leq \Exp{[|C_{ij}|]} \leq \binom{n-2}{d-2} p_{\max} = \frac{\binom{n-2}{d-2}c_0 \log n }{\binom{n-1}{d-1}} \lesssim \frac{\log n}{n}.  
\end{align*}
Therefore, by Lemma \ref{lem:sp-norm-i.i.d matrix} and the fact that $\Exp{[C \mid F]}$ is the zero matrix, we obtain
$$\Exp{[ \twonorm{C} \mid F]} = \Exp[\twonorm{C-\Exp{[C \mid F]}} \mid F ] \leq c' \sqrt{n \frac{\log n}{n}} \leq c' \sqrt{\log n}.$$ 
Substituting this back in \eqref{eq:sp norm final} and using \eqref{eq:conditional-expectation remaining part},
$$\Exp{[\twonorm{W-W^*}]} \leq 2K \Exp[\twonorm{C}] \leq 2K\Exp[\twonorm{C} \mid F ]+o(1) \leq 2Kc' \sqrt{\log n} +o(1) \leq c \sqrt{\log n},$$
for some $c$ that depends on $d$ and $c_0$, concluding the proof.
\end{proof}

\end{document}